\documentclass[reprint,amsmath,amssymb,aps,pra]{revtex4-1}
\usepackage{upgreek}
\usepackage[british]{babel}
\usepackage{dcolumn}
\usepackage{bm}
\usepackage{graphicx}
\usepackage{amsthm}
\theoremstyle{plain}
\newtheorem{theorem}{Theorem}
\newtheorem{proposition}[theorem]{Proposition}
\theoremstyle{remark}
\newtheorem{remark}{Remark}
\theoremstyle{definition}
\newtheorem{definition}{Definition}
\usepackage[mathscr]{euscript}
\usepackage{mathbbol}
\usepackage[sans]{dsfont}
\usepackage[mathscr]{euscript}
\usepackage{mathbbol}
\usepackage[sans]{dsfont}

\newcommand{\Tr}{\operatorname{Tr}}
\newcommand{\rmd}{\mathrm{d}}
\newcommand{\rmi}{\mathrm{i}}
\newcommand{\rme}{\mathrm{e}}
\newcommand{\bbone}{\mathds1}
\newcommand{\norm}[1]{\left\Vert#1\right\Vert}
\newcommand{\abs}[1]{\left\vert#1\right\vert}
\DeclareMathAlphabet{\mathpzc}{OT1}{pzc}{m}{it}
\newcommand{\Rbb}{\mathbb{R}}
\newcommand{\Cbb}{\mathbb{C}}
\newcommand{\Pbb}{\mathbb{P}}
\newcommand{\Qbb}{\mathbb{Q}}
\newcommand{\Ebb}{\operatorname{\mathbb{E}}}

\newcommand{\Gcal}{\mathcal{G}}

\newcommand{\Ical}{\mathcal{I}}
\newcommand{\Jcal}{\mathcal{J}}

\newcommand{\Lcal}{\mathcal{L}}

 \newcommand{\Wcal}{\mathcal{W}}

\newcommand{\Bscr}{\mathscr{B}}

\newcommand{\Fscr}{\mathscr{F}}
\newcommand{\Gscr}{\mathscr{G}}
\newcommand{\Hscr}{\mathscr{H}}

\newcommand{\Kscr}{\mathscr{K}}

\newcommand{\Sscr}{\mathscr{S}}

\newcommand{\Cov}{\operatorname{Cov}}
\newcommand{\Var}{\operatorname{Var}}
\newcommand{\RE}{\operatorname{Re}}
\newcommand{\IM}{\operatorname{Im}}
\begin{document}
\title{Quantum continuous measurements: The stochastic Schr\"odinger equations and
the spectrum of the output}


\author{Alberto Barchielli and   Matteo Gregoratti}

\affiliation{Politecnico di Milano, Department of Mathematics,
    Piazza Leonardo da Vinci 32, 20133 Milano, Italy}
   \altaffiliation[Also at ]{Istituto Nazionale di Fisica Nucleare, Sezione di Milano,
    and Istituto Nazionale di Alta Matematica, GNAMPA }

\begin{abstract}The stochastic Schr\"odinger equation, of classical or quantum type, allows to describe open quantum systems under measurement in continuous time. In this paper we review the link between these two descriptions and we study the properties of the output of the measurement. For simplicity we deal only with the diffusive case. Firstly, we discuss the quantum stochastic Schr\"odinger equation, which is based on quantum stochastic calculus, and we show how to transform it into the classical stochastic Schr\"odinger equation by diagonalization of suitable commuting quantum observables.
Then, we give the a posteriori state, the conditional system state at time $t$ given the output up to that time, and we link its evolution to the classical stochastic Schr\"odinger equation. Moreover, the relation with quantum filtering theory is shortly discussed. Finally,  we study the output of the continuous measurement, which is a stochastic process with probability distribution given by the rules of quantum mechanics. When the output process is stationary, at least in the long run, the
spectrum of the process can be introduced and its properties studied. In particular we
show how the Heisenberg uncertainty relations give rise to characteristic bounds on the
possible spectra and we discuss how this is related to the typical quantum phenomenon
of squeezing. We use a simple quantum system, a two-level atom stimulated by a laser, to discuss the differences between  homodyne and heterodyne detection and to explicitly show squeezing and anti-squeezing of the homodyne spectrum and the Mollow triplet in the fluorescence spectrum.\end{abstract}


\maketitle

\section{Introduction}\label{S1}

A big achievement in the 70's-80's was to show that, inside the modern formulation of
quantum mechanics, based on \emph{positive operator valued measures} and
\emph{instruments} \cite{Kra80,Dav76}, a consistent theory of measurements
in continuous time (\emph{quantum continuous measurements}) was
possible \cite{Dav76,BarLP82,BarL85,Bar86PR,Bel88,BarB91,Hol01}.
Starting from the 80's, two other very flexible and
powerful formulations of continuous measurement theory were developed. The first one is often referred to as \emph{quantum trajectory theory} and it is based on the \emph{the stochastic Schr\"odinger equation} (SSE), a stochastic differential
equation of classical type (commuting noises, It\^o calculus) \cite{Bel88,BarB91,Car93,WisMil93,BarP96,WisM93,Wis96,Mil96,Car08,BarG09,WisM10}. The second formulation is based on quantum stochastic calculus \cite{HudP84,GarC85,Parthas92} and the \emph{quantum SSE} (non commuting noises, Bose fields, Hudson-Parthasarathy
equation) \cite{BarL85,Bar86PR,Bel88,BarP96,GarZ00,Hol01,Bar06,Car08,WisM10}.
The main applications of
quantum continuous measurements are in the photon detection theory in quantum optics
(\emph{direct, heterodyne, homodyne
detection}) \cite{Bar90QO,Car93,WisMil93,WisM93,BarP96,Wis96,Mil96,GarZ00,BarL00,BarP02,BarG09,Bar06,Car08,WisM10}. While the classical SSE gives a differential description of the joint evolution of the observed signal and of the measured system, in agreement with the axiomatic formulation of quantum mechanics, the quantum SSE gives a dilation of the measurement process, explicitly introducing an environment which interacts with the system and mediates the observations.

In this paper we start by giving a short presentation of continuous measurement theory based on the
quantum SSE and we show the equivalence between this approach and the one based on the classical SSE (Secs.\ \ref{sec:SSE} and \ref{sec:CM}). Then we consider the output of a continuous measurement and we develop the theory up to the introduction of its spectrum (Sec.\ \ref{sec:spectrum}), which enables the study of typical and significative applications, see Sec.\ \ref{sec:model}.

We consider only the type of observables relevant for the description of
homodyne/heterodyne detection and we make the mathematical simplification of introducing only bounded
operators on the Hilbert space of the quantum system of interest and a finite number of noises; for the case of unbounded operators see \cite{FagW03,Fag06,CastroB11}.

In Sec.\ \ref{sec:SSE}, first we discuss some typical approximations that give rise to a system/environment interaction described by a quantum SSE, then we give a mathematical meaning to such an equation by introducing the basic ingredients of quantum stochastic calculus on Fock space.

In Sec.\ \ref{sec:CM} we introduce the quantum observables which describe the continuous measurement and we show how to derive the classical SSE and the related stochastic master equation (SME). The key point in the step from the quantum SSE to the classical SSE is the introduction of an Hilbert space isomorphism which diagonalizes a suitable complete set of quantum observables.
We shortly illustrate also the connections of this approach to quantum filtering theory \cite{Bel88,BouVHJ07}.
The classical SSE and the SME give both the probability distribution for the observed output and the \emph{a posteriori state}, the conditional system state given a realization of the output. These equations are driven by classical noises, but, in spite of this, they are fully quantum as they are equivalent to the formulation of continuous measurements based on quantum fields. It is just in the formulation based on SSE and SME that the probabilistic structure of the output current becomes very transparent.

In Sec.\ \ref{sec:spectrum} we
introduce the spectrum of the classical stochastic process which represents the output and we
study the general properties of the spectra of such processes by proving
characteristic bounds due to the Heisenberg uncertainty principle. This bound is one of the evidences that the whole theory of continuous measurements is fully quantum, independently of the adopted formulation.

As an application, in Sec.\ \ref{sec:model} we present the
case of a two-level atom, which is measured in continuous time by detection of its fluorescence light. The spectral analysis of the output can reveal the phenomenon
of squeezing of the fluorescence light, a phenomenon related to the  uncertainty
relations. We use this example also to illustrate the differences between homodyning and heterodyning and between the \emph{spectrum of the squeezing} and the \emph{power spectrum}. Finally we show how Mollow triplet appears in the power spectrum in the case of an intense stimulating laser. Section \ref{sec:concl} contains our conclusions.

\section{The quantum stochastic Schr\"odinger equation}\label{sec:SSE}

We want to introduce the continuous measurement theory of a system $S$ at its higher level, that is by explicitly modelling also the quantum environment which mediates the observation: we observe an environment that is coupled with $S$, thus acquiring direct information on the environment and indirect information on $S$.

Thus we start from the quantum SSE which will be used to define the global evolution of $S$ and the environment. Even if not defined by a regular Schr\"odinger equation, such an evolution will be a proper Hamiltonian evolution, but of a very particular kind, with ``Markovian'' features for $S$. Indeed, the quantum SSE naturally emerges with the typical ``Markovian'' limits which allow to describe an open evolution of $S$ by a quantum dynamical semigroup or to describe a continuous measurement of $S$ by an \emph{instrumental process} \cite{Hol01}.

Therefore a quantum SSE defines an approximated model, which, nevertheless, is fully quantum, and for this reason it can describe very well typical quantum experiments, such as the quantum optical ones or in general quantum continuous measurements.

In this section, first we show some typical approximations which lead to a quantum SSE, then we introduce quantum stochastic calculus (QSC), the mathematical background which gives a meaning to that equation and allows to start with its study.

In the following we shall denote by $\Hscr$ the \emph{system space,} the complex separable Hilbert space associated to the
open quantum system $S$.

\subsection{The physical bases of the quantum SSE}\label{physapprox}
We start by presenting the typical physical approximations which are involved in the use of the quantum SSE \cite{Bar90QO,Bar91,GarC85,GarZ00}.

Our system
$S$ interacts with some quantum Bose fields $\hat b_k(\nu)$, satisfying the the canonical commutation relations (CCR)
\begin{equation}\label{CCRfrequency}
[\hat b_k(\nu),\hat b_i(\nu^{\prime})] = 0, \qquad [\hat b_k(\nu),\hat b_i^\dagger(\nu^{\prime})] =
\delta_{ki}\, \delta(\nu - \nu^{\prime}) .
\end{equation}
The parameter $\nu$ is the energy or the frequency (we are taking $\hbar=1$) of the free field, while $k$ is an additional discrete degree of freedom. These fields can represent, for instance, the electromagnetic field; in this case the index $k$ stays for polarization, (discretized) direction of propagation, and so on \cite{Bar97}.

A generic system--field interaction, linear in the field operators, can be written as
\begin{equation}\label{Hint}
H_{\mathrm {I}} = \sum_k \frac{\rmi}{ \sqrt {2\pi}}
\int_{\Omega_k-\theta_k}^{\Omega_k+\theta_k} \kappa_k(\nu) \bigl[R_k \hat b_k^{\dagger}(\nu) -
R_k^* \hat b_k(\nu)\bigr] \rmd \nu ,
\end{equation}
where the $R_k$ are system operators (acting on $\Hscr$) and the $\kappa_k(\nu)$ are real couplings.  In the optical case, typically the $R_k$ are dipole operators and the rotating-wave approximation is understood. The $\Omega_k$ are resonance frequencies of system $S$ and $2\theta_k$ is the interaction bandwidth.

By working with the Heisenberg equations of motion for system operators, Gardiner and Collet \cite{GarC85} discussed the approximations needed to pass
from the quasi--physical Hamiltonian \eqref{Hint} to a Markovian quantum stochastic evolution. Here we present the same approximations by working with the global unitary evolution of the composed system ``$S$ plus fields'' \cite{Bar90QO}.

\subsubsection{The flat-spectrum approximation}

The first approximation is to take the couplings independent from $\nu$: the \emph{flat-spectrum approximation}. As a constant can always be included in $R_k$, we take $\kappa_k(\nu)=1$.
Then, we pass to \emph{the interaction picture with respect to the free dynamics of the fields}:
\begin{equation}\label{freefieldev}
\hat b_k(\nu) \mapsto \rme^{-\rmi \nu t} \hat  b_k(\nu).
\end{equation}
In this picture, the interaction Hamiltonian becomes
\begin{equation}\label{Hintint}
\tilde H_{\mathrm {I}}(t) = \rmi \sum_k \left[ R_k \tilde b_k^\dagger(t) - R_k^\dagger \tilde b_k(t)
\right] ,
\end{equation}
\begin{equation}\label{ajint}
\tilde b_k(t) := \frac{1}{ \sqrt{2\pi}} \int_{\Omega_k - \theta_k}^{\Omega_k + \theta_k}  \rme^{-\rmi \nu t}
\hat b_k(\nu)\, \rmd \nu .
\end{equation}
 By construction, the field operator $\tilde b_k(t)$ represents  a wave packet with some carrier frequency
$\Omega_k$ and bandwidth $2\theta_k$.

In the interaction picture, the time evolution operator $\tilde U_t$  can be written
as
\begin{equation}\label{Utilde}
\tilde U_t = \overleftarrow{\mathrm {T}} \, \exp \Bigl\{ -\rmi \int_0^t \left[ H_0 + H_{\mathrm{
I}}(s) \right] \rmd s \Bigr\} \,,
\end{equation}
where $H_0$ is the free  Hamiltonian of system $S$ and $\overleftarrow{\mathrm{ T}}$ is the usual time--ordering
prescription (chronological product).

\subsubsection{The  broad-band approximation}
Now, we take the \emph{broad-band approximation}:
$\theta_k \to +\infty$, $\forall k$.
Note that in this limit the energy of the free field becomes unbounded both from above and from below. This approximation is justified when only  energies not far from the resonance frequencies $\Omega_k$ are involved  in the physical process \cite[p.\ 149]{GarZ00}.

From  \eqref{CCRfrequency}  and \eqref{ajint} we obtain that the field operators $b_k(t) := \displaystyle \lim_{\theta_k\to +\infty}  \tilde b_k(t)$ are given by
\begin{equation}\label{ajint2}
b_k(t) = \frac{1}{ \sqrt{2\pi}} \int_{- \infty}^{+ \infty}  \rme^{-\rmi \nu t}
\hat b_k(\nu)\, \rmd \nu
\end{equation}
and satisfy the CCR
\begin{equation}\label{CCR}
\left[b_i(s),b_k^\dagger(t) \right]=\delta_{ik}\delta(t-s), \qquad \left[b_i(s),b_k(t) \right]=0.
\end{equation}
Then the free dynamics of the fields \eqref{freefieldev} gives
\begin{equation}\label{freefieldev2}
b_k(s) \mapsto b_k(s+t),
\end{equation}
so that the argument $t$ in the fields $b_k(t)$ has a double role: it is a field degree of freedom, the conjugate momentum of the free field energy $\nu$, because of \eqref{ajint2}, and it is the time, because $b_k(t)$ is the evolution of $b_k(0)$ at time $t$.
Let us note that Bose fields with delta-commutations in time were already found by Yuen and Shapiro \cite{YuenS78} in their study of the \emph{quasi-monochromatic paraxial approximation} of the electromagnetic field.

Let us take now the limit $\theta_k \to +\infty$, $\forall k$, in the Hamiltonian $\tilde H_{\mathrm {I}}(t)$ and in the evolution $\tilde U_t$. Formally
\[
\tilde H_{\mathrm {I}}(t)
\to \rmi \sum_k \left(R_k b_k^\dagger(t) - R_k^\dagger b_k(t) \right),
\]
which describes a singular interaction, but which is not a proper operator in a Hilbert space because the singular field operator $b_k^\dagger(t)$ is not integrated over $t$. Nevertheless, $\tilde U_t \to U_t$ where
\begin{equation}\label{U-Tprod}
U_t = \overleftarrow{\mathrm{ T}} \exp \biggl\{ \int_0^t \Bigl[-\rmi H_0 + \sum_k
\left(R_k b_k^\dagger(s) - R_k^\dagger b_k(s)\right) \Bigr]\rmd s \biggr\}
\end{equation}
which can be a proper unitary evolution in a Hilbert space thanks to QSC. Indeed, QSC is the mathematical theory that gives a meaning to It\^o-type integrals with respect to the non-commuting noises $\rmd B_k(t)=b_k(t) \rmd t$ and $\rmd B_k^\dagger(t)=b_k^\dagger(t) \rmd t$, that is with respect to \cite[Eq.\ (11.2.24)]{GarZ00}
\begin{equation}\label{Bb}
B_k(t)=\int_0^tb_k(s) \rmd s , \qquad B_k^\dagger(t)=\int_0^tb_k^\dagger(s) \rmd s.
\end{equation}

To find an equation for $U_t$ we can write
\begin{multline}\label{Ut+}
\rmd U_t=U_{t+\rmd t} - U_t \\
{}= \Bigl(\exp \Bigl\{-\rmi H_0\rmd t + \sum_k \left[R_k \rmd B_k^\dagger(t) - R_k^\dagger \rmd B_k(t)\right] \Bigr\} \\ {} - \bbone\Bigr) U_t\,,
\end{multline}
and then we can try a series expansion of the exponential. However, from  \eqref{CCR} and \eqref{Bb} we get
\[
\left[\rmd  B_k(t),\, \rmd B_i^\dagger(t)\right]=\delta_{ki}\, \rmd t.
\]
Due to this fact, in the second order term of the expansion of the exponential surely some of the contributions are of order $\rmd t$. This shows that new mathematical tools are needed to treat the singular interaction appearing in \eqref{U-Tprod}, QSC indeed. In addition, here another peculiarity arises: in the case of quantum fields there exist non unitarily equivalent representations of the CCR \eqref{CCR}. Moreover, which ones of the second order terms are indeed of order $\rmd t$ depends on the representation and this implies that the rules of QSC depend on the representation of the CCR \cite{HudL85}.

In this paper we consider only the representation of the CCR \eqref{CCR} on the Fock space, the one characterized by the existence of the vacuum state. Let us stress that representations not unitarily equivalent to the Fock one describe physically different situations, such as thermal and squeezed input fields \cite{GarC85,GarZ00}.

The final result of these approximations is the quantum SSE \eqref{HPequ}, as we shall see after Theorem \ref{th1}. Note that, however, such an equation is a general evolution model that emerges under many different Markovian limits, not only with the one we have just described. Of course, the system operators $H_0$ and $R_k$ in \eqref{U-Tprod} or \eqref{HPequ} are to be chosen just by looking at the physical context and the approximations that produce the Markovian regime. In the model of Sec.\ \ref{sec:model} we shall show how to represent various dissipative effects by a suitable choice of the operators $R_k$.

\subsection{Quantum stochastic calculus and unitary dynamics}

We introduce now QSC and the Hudson-Parthasarathy equation in the Fock representation. QSC \cite{HudP84} is based on the use of some Bose fields, satisfying the CCR with a Dirac delta in time \eqref{CCR}, that model the environment interacting with an initial system $S$ and play the role of non-commuting noises. By QSC one gives meaning to the quantum stochastic Schr\"odinger equation or
Hudson-Parthasarathy equation \cite{HudP84,Parthas92}. For a short review see \cite[Sec.\ 2]{Bar06} or \cite[Secs.\ 11.1, 11.2]{GarZ00}.
Our aim is to recall the main notions and to fix the notations, not to give a self-contained presentation, which can be found in \cite{Parthas92,Bar06}.

\subsubsection{Fock space}

Let $\Kscr$ be the separable Hilbert space of a bosonic particle and $\Kscr^{\otimes_s n}$  be the ``$n$-particle space'', that is the symmetric part of the tensor product $\Kscr\otimes \cdots \otimes \Kscr$, $n$ times. Then, the direct sum $\displaystyle\Gamma(\Kscr)= \Cbb \oplus \sum_{n=1}^{+\infty} \Kscr^{\otimes_s n}$ is the \emph{symmetric} (or bosonic) \emph{Fock space} over $\Kscr$.
In this context a \emph{coherent vector}
$e(f)$,  $f \in \Kscr$, is the vector in $\Gamma(\Kscr)$ given by
\begin{equation}\label{expV}
e(f) :=\rme^{
-\frac{1}{2}\,\norm{f}^2}\left(1,f,\frac{f\otimes f}{\sqrt{2!}},\ldots,\frac{f^{\otimes n}}{\sqrt{n!}}, \ldots \right).
\end{equation}
Note that $e(0)$ represents the vacuum state and that $\langle e(g)|e(f)\rangle = \exp\left\{ -\frac 1 2 \norm{f}^2 -\frac 1 2 \norm{g}^2 +\langle g|f\rangle\right\}$; in particular, the coherent vectors are normalized. Moreover, they are all linearly independent and their linear combinations are dense in $\Gamma(\Kscr)$. Then, an important property of the Fock spaces is that the action on the coherent vectors uniquely determines a densely defined  linear operator.
We are interested in  $\Kscr= L^2(\mathbb{R})\otimes\Cbb^d =
L^2({\mathbb R}; \Cbb^d)$ and we denote by $\Gamma\equiv \Gamma\big(L^2({\mathbb R}; \Cbb^d)\big)$ the symmetric Fock space over the one-particle space $L^2({\mathbb R}; \Cbb^d)$.

\subsubsection{Factorization properties of the Fock space}
A general property of symmetric Fock spaces is that, when the one-particle space is given by a direct sum ($\Kscr=\Kscr_1\oplus \Kscr_2$), then the factorization property $\Gamma(\Kscr_1\oplus \Kscr_2)= \Gamma(\Kscr_1)\otimes\Gamma( \Kscr_2)$ holds.

In our set up, for every time interval $A$, let us denote by $\Gamma[A]\equiv \Gamma\big(L^2(A; \Cbb^d)\big)$ the symmetric Fock space over $ L^2(A; \Cbb^d)$; in particular, we have $\Gamma=\Gamma[\Rbb]$. Then,
for any $s<t$, we have $L^2(\Rbb; \Cbb^d)= L^2\big((-\infty,s); \Cbb^d\big)\oplus L^2\big((s,t); \Cbb^d\big)\oplus L^2\big((t,+\infty); \Cbb^d\big)$ and
\begin{equation}\label{FockFactorization}
\Gamma[\mathbb{R}] = \Gamma\big[(-\infty,s)\big]\otimes\Gamma\big[(s,t)\big]\otimes\Gamma\big[(t,+\infty)\big].
\end{equation}

Moreover, each space $\Gamma[A]$ can be identified with a subspace of the full Fock space $\Gamma[\Rbb]$ by taking the tensor product of a generic vector in $\Gamma[A]$ with the vacuum of $\Gamma[\Rbb\setminus A]$. Then, for every $f\in L^2(\Rbb; \Cbb^d)$, we have the identification
\[
e(f|_A)\in \Gamma[A] \mapsto e(1_A f)\in \Gamma[\Rbb].
\]
We are denoting by $1_A(\cdot)$ the indicator function of the set $A$ and by $f|_A$ the restriction of the function $f$ to the set $A$. With an abuse of notation we write
\[
e(f)=e\big(1_{(-\infty,s)}f\big)\otimes e\big(1_{(s,t)}f\big)\otimes e\big(1_{(t,+\infty)}f\big).
\]
In particular, $e\big(1_{(s,t)}f\big)$ can represent a vector in $\Gamma[\Rbb]$ or in $\Gamma[(s,t)]$ and we have the identification $e\big(1_{(s,t)}f\big)= e(0) \otimes e\big(1_{(s,t)}f\big) \otimes e(0)$.

\subsubsection{Bose fields}

Let $\{z_k,\ k\geq 1\}$ be the canonical basis in $\Cbb^d$ and for any $f\in
L^2(\mathbb{R};\Cbb^d)$ let us set $f_k(t):= \langle z_k|f(t)\rangle_{\Cbb^d}$.

Then, we define two families of mutually adjoint operators, the \emph{annihilation} and \emph{creation processes}, by their actions on the coherent vectors:
\begin{gather*}
B_k(t)\,e(f)= \int_0^t f_k(s) \,\rmd  s \, e(f)\,, \\
\langle e(g)| B_k^\dagger(t)e(f)\rangle = \int_0^t \overline{ g_k(s)}\, \rmd  s \, \langle
e(g)| e(f)\rangle.
\end{gather*}
The overline denotes the complex conjugation.

For $t>0$, the annihilation and creation processes are \emph{adapted}, in the sense that they factorizes, with respect to \eqref{FockFactorization}, as
\begin{equation*}
B_k^{(\dagger)}(t)=\bbone_{(-\infty,0)}\otimes B_k^{(\dagger)}(t)\otimes\bbone_{(t,+\infty)},
\end{equation*}
and they satisfy the integrated form of the CCR, namely
\begin{gather}\label{B_CCR}
[B_k(t),B_l^\dagger(s)]=\delta_{kl}\; t\wedge s, \\ \notag
[B_k(t),B_l(s)]=0,
\qquad [B_k^\dagger(t),B_l^\dagger(s)]=0;
\end{gather}
$t\wedge s$ is the minimum between $t$ and $s$ and $\bbone_A$ is the identity operator on $\Gamma[A]$.

By introducing also the ``field densities\rq{}\rq{} $b_k(t)$ by
\begin{equation}\label{nihil}
b_k(t)\,e(f)= f_k(t)\, e(f) \qquad \forall f\in L^2({\mathbb R}; \Cbb^d),
\end{equation}
we get that all of them annihilate the vacuum and that, together their formal adjoints, they satisfy the CCR \eqref{CCR} and that the annihilation and creation processes are nothing but the integrals \eqref{Bb} of these densities.

\subsubsection{Temporal modes and Weyl operators.}\label{TmWo}

The free evolution in the Fock space is represented by the left shift in $\Gamma$
\begin{equation*}
\Theta_t\,e(f)=e(\theta_tf),\quad \big(\theta_tf\big)(s)=f(s+t).
\end{equation*}
Then, coherently with \eqref{freefieldev2}, the action of the shift on the fields is given by
\begin{equation}\label{EVt}
\Theta_t^\dagger\,b_k(s)\,\Theta_t=b_k(s+t),
\end{equation}
wich is the free evolution \eqref{freefieldev2}.

If we take a function $g\in L^2(\Rbb)$ we can define the annihilation operator
\begin{equation}\label{c_k}
c_k(g):= \int_{-\infty}^{+\infty} \overline{g(t)}\, b_k(t)\, \rmd t.
\end{equation}
By Eq.\ \eqref{nihil}, its action on the coherent vectors is given by
\[
c_k(g)\,e(f)= \int_{-\infty}^{+\infty}  \overline{g(t)}\,f_k(s) \,\rmd  s \, e(f)\equiv \langle g| f_k \rangle_{L^2(\Rbb)} \, e(f).
\]

If we take a complete orthonormal system $g^i$, $i=1,2,\ldots$, in $L^2(\Rbb)$, we can define the annihilation operators $c_k(g^i)$.
Together with their adjoint operators, they satisfy the usual CCR.
We can say that the upper index $i$ denotes the \emph{temporal modes}, while the lower index $k$ denotes the polarization/spatial modes.

An important technical tool is represented by the \emph{Weyl operators} $\Wcal(q)$, $q\in L^2({\mathbb R}; \Cbb^d)$, the unitary operators defined by: $ \forall f\in L^2({\mathbb R}; \Cbb^d)$,
\[
\Wcal(q)e(f)= \exp\left\{\rmi \IM \langle f|q\rangle_{L^2({\mathbb R}; \Cbb^d)}\right\} e(f+q);
\]
this is nothing but the \emph{displacement operator} for the field. By using the notation \eqref{c_k} we can write
\begin{equation}\label{Wcal}
\Wcal(q)= \exp\left\{\sum_k \left(c_k^\dagger(q_k) - \text{h.c.}\right)\right\},
\end{equation}
while, by using the discrete modes introduced above, we have
\[
\Wcal(q)= \exp\left\{\sum_{k i} \left(\langle g^i|q_k \rangle_{L^2(\Rbb)} c_k^\dagger(g^i) - \text{h.c.}\right)\right\}.
\]
By h.c.\ we denote the Hermitian conjugate operator.

\subsubsection{The Hudson-Parthasarathy equation}
Now we want to couple the system $S$ with the fields by constructing a unitary evolution of the composite system in $\Hscr\otimes \Gamma$. When convenient, an operator $X$ on $\Hscr$ (resp. $Y$ on $\Gamma$) is identified with $X\otimes \bbone_\Gamma$ on $\Hscr\otimes \Gamma$ (resp. $\bbone_\Hscr\otimes Y$).

By defining integrals of It\^o type with respect to the increments of the quantum processes $B_k$, $B_k^\dagger$, it is possible to construct adapted operator processes on $\Hscr\otimes \Gamma$ and to develop a quantum stochastic calculus, whose rules are summarized, at a heuristic level, by the quantum It\^o table
\begin{subequations} \label{Itable}
\begin{gather}
\rmd  B_k(t)\, {\rmd } B_l^{\dagger}(t) = \delta_{kl}\, {\rmd } t,   \qquad \rmd  B_k^\dagger(t)\, {\rmd } B_l(t) =0 ,
\\ \rmd  B_k(t)\, {\rmd } B_l(t) =  0,
\qquad
\rmd  B_k^\dagger(t)\, {\rmd } B_l^{\dagger}(t) = 0, \\ \rmd  B_k^\dagger(t)\, {\rmd }t=0, \qquad  \rmd  B_k(t)\, {\rmd }t=0, \qquad (\rmd t)^2=0.
\end{gather}
\end{subequations}
Let us stress that these multiplication rules do not depend on the field  state in the Fock space, but they depend on the representation. Indeed, thermal and squeezed representations, which describe different physical situations, have different It\^o tables \cite[Eqs.\ (5.3.52), (10.2.38)]{GarZ00}.

We can now introduce  the quantum stochastic Schr\"odinger equation or Hudson-Parthasarathy equation \cite{HudP84,Fri85,Parthas92}.
\begin{theorem}[Hudson and Parthasarathy]\label{th1}
Let $H_0$, $R_k$, $k,l=1,\ldots,d$, be bounded operators on $\Hscr$ such that $H_0^\dagger =H_0$. We set also
\begin{equation}\label{K+H}
 K:=-\rmi H_0 -
\frac{1}{2} \sum_k R_k^{\dagger}R_k.
\end{equation}
Then, the quantum stochastic differential equation
\begin{equation}
\rmd U_t = \biggl\{ \sum_k R_k \,\rmd B_k^\dagger(t)
- \sum_{k} R_k^\dagger\,\rmd B_k(t) +
K\,\rmd t \biggr\}  U_t,\label{HPequ}
\end{equation}
with the initial condition $U_0 =\bbone$, has a unique solution, which is a strongly
continuous adapted family of unitary operators on $\Hscr\otimes \Gamma$. Moreover, the family of unitary operators $\Theta_t\,U_t$, $t\geq 0$, and $U^\dagger_{|t|}\,\Theta^\dagger_{|t|}$, $t\leq 0$, is a strongly continuous unitary group.
\end{theorem}

Note that, if we take our limit dynamics \eqref{U-Tprod} of Sec.\ \ref{physapprox}, and we compute the differential $\rmd U_t$ in expression \eqref{Ut+} by expanding the exponential with the It\^o table \eqref{Itable}, we get indeed Eq.\  \eqref{HPequ}. So,
the unitary operators $U_t$  represent the system-field
dynamics in the interaction picture with respect to the free field evolution.

Then, for $t\geq0$, the dynamics in the Schr\"odinger picture is the unitary group $\rme^{-\rmi H_{\mathrm{TOT}}t}=\Theta_t\,U_t$, whose Hamiltonian $H_{\mathrm{TOT}}$ is a singular perturbation of the unbounded generator of $\Theta_t$ \cite{Greg00,Greg01}. Roughly speaking, the system $S$ absorbs or emits bosons instantaneously;
then, the emitted bosons are carried away by their free dynamics and never come back.

Note that the interaction picture with respect to the free field dynamics coincides with the Schr\"odinger picture when only  reduced system states and observables are considered.

\subsection{The reduced dynamics of the system}

The states of a quantum system are represented by statistical operators, positive trace-class
operators with trace one; let us denote by $\Sscr(\Hscr)$ the set of statistical operators on
$\Hscr$. For every composed state $\Sigma$ in $\Sscr(\Hscr\otimes\Gamma)$, the partial trace $\Tr_\Gamma$ (resp.\ $\Tr_\Hscr$) with respect to the field (resp.\ system) Hilbert space gives the reduced system (resp.\ field) state $\Tr_\Gamma\Sigma$ in $\Sscr(\Hscr)$ (resp.\ $\Tr_\Hscr\Sigma$ in $\Sscr(\Gamma)$).

\subsubsection{The initial state and the reduced states}\label{sec:in_st}
As initial state of the composed system ``$S$ plus fields'' we take
$\rho\otimes \varrho_\Gamma(f)\in\Sscr(\Hscr\otimes\Gamma)$, where $\rho\in \Sscr(\Hscr)$ is generic and
$\varrho_\Gamma(f)$ is a coherent state, $\varrho_\Gamma(f):= |e(f)\rangle \langle e(f)|$. Then, the system-field state at time $t$, in the field interaction picture, is
\begin{equation}\label{initial_s}
\Sigma_f(t):= U_t \left(\rho\otimes
\varrho_\Gamma(f)\right)U_t^\dagger.
\end{equation}
The \emph{reduced system state} and the \emph{reduced field state} are
\begin{equation}\label{redst}
\eta_t:=\Tr_\Gamma \left\{\Sigma_f(t)\right\}, \qquad \Pi_f(t):=\Tr_\Hscr\left\{\Sigma_f(t)\right\}.
\end{equation}

\subsubsection{The master equation}
One of the main properties of the Hudson-Parthasarathy equation is that, with the initial
state introduced above, the reduced dynamics of system $S$ exactly obeys a quantum master
equation \cite{HudP84,Parthas92,Bar06}. Indeed, we get
\begin{equation}\label{masteq}
\frac{\rmd \ }{\rmd t}\, \eta_t=\Lcal(t)[\eta_t],
\end{equation}
where the Liouville operator $\Lcal(t)$ turns out to be
\begin{align}\notag
\Lcal(t)[\rho]=& -\rmi\left[H_0+ H_f(t),\,\rho
\right]
\\ {} &+
\sum_k \left(R_k\rho R_k^\dagger-\frac 1 2 R_k^\dagger R_k\rho -\frac 1 2\rho R_k^\dagger R_k\right), \label{Lop}
\end{align}
\begin{equation}\label{hatH}
H_f(t):=   \rmi\sum_k\overline{f_k(t)}R_k-\rmi\sum_k f_k(t)R_k^\dagger.
\end{equation}
Therefore, $S$ is an open system, as it interacts with the fields in $\Gamma$, and its evolution turns out to be Markovian thanks to the properties of the interaction and of the choice of a coherent state as initial state of the environment. Note that the dynamics \eqref{masteq} depends not only on the global evolution \eqref{HPequ} but also on the initial state of the environment $\varrho_\Gamma(f)$.

It is useful to introduce also the evolution operator from $s$ to $t$ by
\begin{equation}\label{propagator}
\frac{\rmd \ }{\rmd t}\,\Upsilon(t,s)=\Lcal(t)\circ \Upsilon(t,s),\qquad
\Upsilon(s,s)=\bbone.
\end{equation}
With this notation we have $ \eta_t=\Upsilon(t,0)[\rho] $.

\section{Continuous monitoring}\label{sec:CM}
The connections among quantum stochastic calculus, quantum Langevin equations and input and output fields were developed by Gardiner and Collet in \cite{GarC85}.
Then, in \cite{Bar86PR} these notions were connected to the unitary evolution \eqref{HPequ} and to continuous measurements. Indeed, another fundamental property of the Hudson-Parthasarathy equation is that it allows for a fully quantum description of a continuous measurement of the system $S$: the measurement is obtained by detecting the bosons that have been emitted by $S$. Of course such a measurement acquires information on both $S$ and the detected bosons.

\subsection{Input and output fields}

Let us call ``input fields'' the fields $B_k(t)$, $B_k^\dagger(t),\ldots$ when they are considered as operators in interaction picture at time $t$, with respect to $\Theta_t$, and let us call ``output fields'' the same fields in the Heisenberg picture:
\begin{equation}\label{out1}
B_{\,k}^{\mathrm{out}}(t):= U_t^\dagger B_k(t) U_t
\end{equation}
and a similar definition for $B_{\,k}^{\mathrm{out}\,\dagger}(t)$. By the properties of the Fock space $\Gamma$ and of the unitary operators $U_t$, it is possible to prove that
\begin{equation}\label{T>t}
B_{\,k}^{\mathrm{out}}(t)= U_T^\dagger B_k(t) U_T\,, \qquad \forall T\geq t.
\end{equation}
This equation is of fundamental importance and it immediately implies that the output fields satisfy the same commutation rules of the input fields, for instance the CCR \eqref{B_CCR}: the output fields remain Bose free fields.
By applying the formal rules of QSC \eqref{Itable}, we can express the output fields as the quantum stochastic integrals \cite{Bar86PR}
\begin{equation}\label{Bout}
B_{\,k}^{\mathrm{out}}(t) = B_k(t) +\int_0^t
U_s^\dagger R_kU_s\, \rmd s;
\end{equation}
$B_{\,k}^{\mathrm{out}\,\dagger}(t)$ is given by the adjoint expression.

\subsection{The field observables}

The key point of the theory of continuous measurements is to consider field observables
represented by time dependent, commuting selfadjoint operators in the Heisenberg
picture \cite{BarL85,Bar86PR,Bar06}. Being commuting at different times, these observables represent outputs
produced at different times which can be obtained in the same experiment. Here, the observables we consider are some field quadratures. Let us start by introducing the selfadjoint operators
\begin{equation}\label{quadrature}
Q(t;\vartheta,h)= \rme^{-\rmi\vartheta}\int_0^t  h(s)\,\rmd B_1^\dagger(s) +
\textrm{h.c.}, \qquad  t\geq 0;
\end{equation}
the phase $\vartheta\in (-\pi,\pi]$ and the function $h$, with $\abs{h(t)}=1$, are fixed.

The operators \eqref{quadrature} have to be interpreted as linear combinations of
the formal increments $\rmd B_1^\dagger(s) $, $\rmd B_1(s) $ which represent field operators
in the
interaction picture.
The corresponding operators in the Heisenberg picture are
\begin{align}\notag
Q^{\mathrm{out}}(t;\vartheta,h):&= U_t^\dagger Q(t;\vartheta,h)U_t \\ {}&=
U_T^\dagger Q(t;\vartheta,h)U_T, \quad\forall T\geq t, \label{Qout}
\end{align}
where the second equality follows from Eq.\ \eqref{T>t}. These ``output'' quadratures are our
observables.

When ``field 1'' represents the electromagnetic field, a physical realization of a measurement
of the observables \eqref{Qout} is implemented by what is called \emph{balanced heterodyne/homodyne
detection} \cite{ShYM79,YuS80,YuenC83}, \cite[Sec.\ 8.4.4]{GarZ00}. The light emitted by the system in the ``channel 1'' interferes with an intense
laser beam represented by the wave $h$, the \emph{local oscillator}; as $|h|=1$, it represents only the phase of the local oscillator wave. The description of the
apparatus and its formalization in mathematical terms  is given in \cite[Sec.\ 3.5]{Bar06}.

Each quadrature $Q^{\mathrm{out}}(t;\vartheta,h)$ is observed at the corresponding time $t$ and it regards those bosons in ``field 1'' which have eventually interacted with $S$
between time $0$ and time $t$, so it can be interpreted as an indirect measurement performed on the system $S$.

By using CCR, one can check that the operators \eqref{quadrature} commute: $[Q(t;\vartheta,h),Q(s;\vartheta,h)]=0$.
The important point is that, thanks to Eq.\ \eqref{Qout}, these operators commute for different times
also in the Heisenberg picture:
\begin{equation}\label{snd}
[Q^{\mathrm{out}}(t;\vartheta,h),Q^{\mathrm{out}}(s;\vartheta,h)]=0.
\end{equation}
Therefore, the observables $Q^{\mathrm{out}}(t;\vartheta,h)$, $t\geq0$, can be jointly measured for every interaction \eqref{HPequ}. The output is a (random) number at every time $t$, that is a signal depending on time, a stochastic process, which is the result of a continuous indirect monitoring of the system $S$. Its probability distribution is given by the usual postulates of quantum mechanics trough the joint diagonalization of the operators $Q^{\mathrm{out}}(t;\vartheta,h)$. Actually, always thanks to Eq.\ \eqref{Qout}, it will be enough to jointly diagonalize the operators $Q(t;\vartheta,h)$.

Let us stress that quadratures of type \eqref{quadrature} with different phases $\vartheta$ and
$h$ functions represent incompatible observables, because they do not commute but satisfy
\[
[Q(t;\vartheta,h),Q(s;\varphi,g)]=2\rmi \int_0^{t\wedge s} \rmd r \, \IM \left(\rme^{\rmi\left(\vartheta- \varphi\right)} \overline{h(r)}\, g(r)\right).
\]
Note that for $g=h$ we get
\begin{equation}\label{2quadrature}
[Q(t;\vartheta,h),Q(s;\varphi,h)]=2\rmi \left(t\wedge s\right)\sin \left(\vartheta- \varphi\right),
\end{equation}
and for $\varphi=\vartheta$ they commute as anticipated.

Let us note that the operator $Q^{\mathrm{out}}(t;\vartheta,h)$ involves the whole time interval $[0,t]$ and has to be interpreted as cumulated output. The \emph{instantaneous output current} is represented by its formal time derivative $\hat I^{\mathrm{out}}(t):=\dot Q^{\mathrm{out}}(t;\vartheta,h)$. From \eqref {Bb}, \eqref{T>t}, \eqref{Bout}, \eqref{Qout} we get
\begin{equation}\label{hatI}
\hat I^{\mathrm{out}}(t)= \rme^{\rmi \vartheta}\, \overline{h(t)}\left( b_1(t) +U_t^\dagger R_1 U_t\right)
+ \text{h.c.}
\end{equation}

\subsection{The stochastic representation}\label{QtoC}

The commuting selfadjoint operators \eqref{quadrature} have a joint projection valued
measure (pvm) $E_\vartheta^h$; by Born rule, it gives the  probability distribution for the output of the continuous measurement. Moreover,    via the partial trace on the fields, $E_\vartheta^h$ gives also the \emph{instruments} describing the transformations of $S$ from time 0 to an arbitrary time $t$, conditioned on the information acquired up to time $t$.
Furthermore, via joint diagonalization and conditioning, the pvm $E_\vartheta^h$ even gives the stochastic evolution of the conditional state $\rho_t$ (or \emph{a posteriori state}),  the state of $S$ at time $t$ given the observed signal from time 0 to time $t$. This evolution turns out to satisfy a stochastic differential equation (SSE or SME), with classical driving noises. The introduction of such stochastic evolution equations for the conditional state was an achievement of the \emph{quantum filtering theory} \cite{Bel88,BelS92,Bel89,Bel94,BouVHJ07,BarB91}.

The passage from the formulation with quantum fields and Hudson-Parthasarathy equation to the one based on classical stochastic differential equations can be done by different techniques. The technique based on the use of isomorphisms between the Fock space and the Wiener space  is very powerful and clear; here we present a variant of the construction given in \cite{BarP96}.

Let us note that the observation we consider is not complete, because it regards only field 1 and involves only positive times. To make unique the isomorphism which diagonalizes the self-adjoint operators \eqref{quadrature}, we need to add fictitious observations, involving quadratures of the fields $2 ,\ldots,d$ too. So, we take a function $\ell \in L^\infty(\Rbb; \Cbb^d)$ such that
\begin{equation}\label{f_ell} \begin{cases}
\abs{\ell_k(t)}=1,  \qquad &\forall t\in \Rbb, \ \forall k,  \\  \ell_1(t)= \rme^{-\rmi \vartheta } h(t),  &\forall t\geq 0.
\end{cases}\end{equation}
Then, we introduce the field quadratures: for $ k=1,\ldots,d$,
\begin{equation}\label{+quadr}
Q_k(t):= \int_0^t \ell_k(s)\rmd B^\dagger_k(s) + \mathrm{h.c.}
\end{equation}
We use this definition for positive and negative times by taking the convention $\int_0^t=-\int_t^0$ for a negative $t$.
These quadratures form a complete set of compatible observables on the Fock space $\Gamma$. Note that $Q_1(t)= Q(t;\vartheta, h)$. In the following subsection we jointly diagonalize all the observables \eqref{+quadr} by introducing an explicit isomorphism between Fock and Wiener spaces.

\subsubsection{Spectral representation on the Wiener space}

Fixed the functions $\ell_1,\ldots,\ell_d$, that is the field quadratures \eqref{+quadr}, we look for a probability space $(\Omega,\Fscr,\Qbb)$, a unitary operator $J:\Gamma[L^2(\Rbb;\Cbb^d)]\to L^2(\Omega,\Fscr,\Qbb)$ (the Hilbert space of the complex square integrable random variables on the given probability space), and a family of random variables $W_k(t)$ on $\Omega$ such that
\begin{equation}\label{WSR}
\bigl(JQ_k(t)\Psi\bigr)(\omega)=W_k(t;\omega)\,\bigl(J\Psi\bigr)(\omega),
\end{equation}
for all $t$, $k$, for almost all $\omega$, and for all $\Psi $ in the domain of the selfadjoint operator $Q_k(t)$. This means that each $Q_k(t)$ is represented in $ L^2(\Omega,\Fscr,\Qbb)$ as the multiplication operator by $W_k(t)$. We can get such a joint diagonalization on the space of the canonical representation of the Wiener process; a short presentation of the canonical Wiener process is given  in \cite[Secs.\ A.2.4, A.2.6]{BarG09}.

\begin{remark}[The Wiener space]\label{rem:Wiener}
Let $\Omega=C_0(\Rbb;\Rbb^d)$ be the space of the continuous functions $\omega: \Rbb \to \Rbb^d$ such that $\omega(0)=0$. We define the $d$-dimensional process $W(t):\Omega\to\Rbb^d$, $t\in\Rbb$, by $W(t,\omega)=\omega(t)$ and we denote by $\Fscr$ the smallest $\sigma$-algebra of subsets of $\Omega$ for which these functions $W(t)$ are  measurable: $\Fscr=\sigma\big(W(t):t\in\Rbb\big)$. Then, there exists a unique probability measure $\Qbb$ on the measurable space $(\Omega, \Fscr)$, the \emph{Wiener measure}, such that the processes $W_k(t)$, $W_k(-t)$, $t\geq 0$, $k=1,\ldots, d$ are $2d$ independent standard Wiener processes.
Moreover, for positive times we introduce the natural filtration $(\Fscr_t)_{t\geq 0}$ of the process $W$: $\Fscr_t=\sigma(W(s):s\in [0,t])$. Finally, the Hilbert space $L^2(\Omega,\Fscr,\Qbb)$ is called \emph{Wiener space}.
\end{remark}
Let us also recall that, if $\phi, \psi\in  L^2(\Omega,\Fscr,\Qbb)$, then their inner product is given by the $\Qbb$-expectation $\Ebb_\Qbb$:
\[
\langle \psi|\phi\rangle= \Ebb_\Qbb[\overline{\psi}\, \phi] =\int_\Omega \overline{\psi(\omega)}\, \phi(\omega)\Qbb(\rmd \omega).
\]

\begin{definition}[The isomorphism $J$]
Let $J:\Gamma[L^2(\Rbb;\Cbb^d)]\to L^2(\Omega,\Fscr,\Qbb)$ be the linear operator defined by: \quad $\forall g\in L^2({\mathbb R}; \Cbb^d)$,
\begin{align}\notag
J \, e(g) =&\exp\left\{\sum_{k=1}^d\int_{-\infty}^{+\infty}  \overline{\ell_k(s)}\, g_k(s)\,\rmd W_k(s)\right\}\\ \notag
{}&\times{} \exp\left\{- \frac{1}{2}\sum_{k=1}^d\int_{-\infty}^{+\infty}  \left( \overline{\ell_k(s)}\, g_k(s)\right)^2\,\rmd s\right\}
\\ {}&\times \exp\left\{- \frac{1}{2}\sum_{k=1}^d\int_{-\infty}^{+\infty}  \abs{g_k(s)}^2\,\rmd s\right\}.\label{isoJ}
\end{align}
In particular we have
\[
J\,e(0)=1,
\qquad
J\, e(1_{(0,t)}f)\in L^2(\Omega,\Fscr_t,\Qbb).
\]
\end{definition}
The operator $J$ turns out to be an  isomorphism and it realizes the representation \eqref{WSR}: $J\,Q_k(t)\,J^{-1}=W_k(t)$, i.e.\ the field quadratures are mapped into the operators  ``multiplication by the Wiener processes''.
Because the isomorphism $J$ jointly diagonalizes all the observables \eqref{+quadr}, then their joint pvm on the Fock space is $J^{-1}1_AJ$, $\forall A \in \Fscr$.

\paragraph{The distribution of the output}
Let us restrict now to the observed quadrature \eqref{quadrature}; the $\sigma$-algebra  $\Gscr_\infty=\sigma(W_1(t):0\leq t <+\infty)$ is the space of all the events regarding our observables $Q(t;\vartheta,h)$, $t\geq 0$.
Then, the joint pvm $E_\vartheta^h$ of the observed quadratures  is defined on the measurable space $(\Omega,\Gscr_\infty)$ by
\begin{equation}\label{pvm}
E_\vartheta^h(G)=J^{-1}1_GJ,\qquad \forall G\in\Gscr_\infty.
\end{equation}
Finally, we get the distribution of the output. By setting
$\Gscr_t=\sigma(W_1(s):s\in [0,t])$,   then, $\Gscr_t\subset \Gscr_\infty$, is the space of the observed events up to time $t$, associated to the observables $Q(s;\vartheta,h)$ for times from 0 to $t$, and, according to the usual rules of quantum mechanics, the probabilities of such events are given by
\begin{equation}\label{P(G)}
\Pbb_{\rho,t}^{\vartheta,h}(G)= \Tr\left\{\bigl( \bbone_\Hscr\otimes E_\vartheta^h(G)\bigr)\Sigma_f(t)\right\},\ \forall G\in \Gscr_t.
\end{equation}

Note that, when the field state is the vacuum and there is no interaction between system $S$ and the fields, this probability reduces to $ \Pbb_{\rho,t}^{\vartheta,h}(G)= \langle e(0)| E_\vartheta^h(G)\,e(0) \rangle= \Ebb_\Qbb[1_G]=\Qbb(G)$. This means that in this case the quadratures \eqref{quadrature} are distributed as a standard Wiener process.

Let us stress that the pvm \eqref{pvm} depends
on the parameters $\vartheta$ and $h$ defining the quadrature \eqref{quadrature}; these parameters are contained in the definition of the isomorphism $J$ \eqref{isoJ}. On the contrary, the choice of the
trajectory space (the measurable space $(\Omega, \Gscr_\infty)$) and the definition of $W_1$ are
independent of $\vartheta$ and $h$. With respect to the time dependence, the physical probabilities \eqref{P(G)} are \emph{consistent}, i.e.
\begin{equation}\label{ts_cons}
0\leq s\leq t, \quad G\in\Gscr_s \ \Rightarrow \ \Pbb_{\rho,t}^{\vartheta,h}(G)=\Pbb_{\rho,s}^{\vartheta,h}(G).
\end{equation}
This result is due to the factorization property \eqref{FockFactorization} of the Fock space and to the localization properties of $U_t$ \cite[Theor.\ 2.3]{Bar06}, which imply
$
U_t^\dagger \bigl( \bbone_\Hscr\otimes E_\vartheta^h(G)\bigr)U_t=U_s^\dagger \bigl( \bbone_\Hscr\otimes E_\vartheta^h(G)\bigr)U_s$ for $0\leq s\leq t$ and $G\in \Gscr_s$, cf.\ Eq.\ \eqref{Qout}.

A more detailed study of the statistical properties of the output needs the introduction of the characteristic operator (Sec.\ \ref{co,pm}).

\subsubsection{The instruments}\label{sec:inst}
The observation of the emitted field can be interpreted as an indirect measurement on the system $S$ and this is formalized by the concept of \emph{instrument} \cite{Kra80,Dav76}. The family of instruments $\Ical_t$, $t>0$, describing our measure is defined by: $\forall G\in \Gscr_t$, $\forall\tau \in \Sscr(\Hscr)$,
\begin{equation}\label{inst}
\Ical_t(G)[\tau]= \Tr_\Gamma\left\{\bigl( \bbone_\Hscr\otimes E_\vartheta^h(G)\bigr)U_t \bigl(\tau\otimes \varrho_\Gamma(f)\bigr)U_t^\dagger\right\}.
\end{equation}
For $\tau=\rho$, the initial system state, Eq.\ \eqref{inst} gives the unnormalized state of $S$ at time $t$ conditioned on the information that the values of the signal in the time interval from 0 to $t$ were in $G$. Of course we have
\begin{equation}\label{P+inst}
 \Tr\left\{\Ical_t(G)[\rho]\right\}=\Pbb_{\rho,t}^{\vartheta,h}(G) ,
\end{equation}
while the normalized conditioned state is given by $\Ical_t(G)[\rho]$ divided by its trace \eqref{P+inst}.

Let us remark that
\begin{equation}\label{etaI}
\eta_t=\Ical_t(\Omega)[\rho],
\end{equation}
so that the system reduced state at time $t$ in the case of no observation ($\eta_t$) coincides with the so called \emph{a priori state} ($\Ical_t(\Omega)[\rho]$), that is the system state at time $t$ in the case of observation performed but not taken into account. This is in agreement with our rough picture of the measurement process: we observe fields which have already interacted with system $S$ and which will never interact again with it. This means that we acquire information on $S$, as we have  $\Ical_t(G)[\rho]\neq  \Pbb_{\rho,t}^{\vartheta,h}(G) \eta_t$, but we do not add any perturbation on its evolution as we have $\Ical_t(\Omega)[\rho]= \eta_t$.

\paragraph{The a posteriori states}

Now we want to introduce $\rho_t$, the state of $S$ at time $t$ conditioned on the whole information supplied by our indirect measurement between time 0 and time $t$, that is by the signal produced by the measurement of $Q(s;\vartheta,h)$ for $ s\in [0,t]$. Therefore, $\rho_t$ has to be a random state depending on the output $W_1(s)$, $0\leq s\leq t$, that is a random state measurable with respect to $\Gscr_t$; in other terms, we have the functional dependence $\rho_t(\omega)= \rho_t\big(\omega_1(s),\, 0\leq s \leq t)$. Such a state is called \emph{a posteriori state} and it is determined by the initial state $\rho$ and by the instrument $\Ical_t$: it is the unique $\Gscr_t$-measurable random state such that
\begin{equation}\label{aps}
\Ical_t(G)[\rho]= \int_G \rho_t(\omega)\,\Pbb_{\rho,t}^{\vartheta,h}(\rmd\omega),  \qquad \forall G\in \Gscr_t.
\end{equation}
The definition of a posteriori state is not linked only to measurements in continuous time, but it has been introduced for a generic instrument \cite{Ozawa}.

As we have a reference probability $\Qbb$ on the output space $(\Omega,\Gscr_t)$ we can equivalently look for the \emph{unnormalized a posteriori state} $\sigma_t$, the unique $\Gscr_t$-measurable random positive operator such that
\begin{equation}\label{nnap}
\Ical_t(G)[\rho]= \int_G \sigma_t(\omega)\,\Qbb(\rmd\omega),  \qquad \forall G\in \Gscr_t.
\end{equation}
Then, $\Tr\{\sigma_t\}$ is the probability density of $\Pbb_{\rho,t}^{\vartheta,h}$ with respect to $\Qbb$ and we have $\rho_t=\sigma_t/\Tr\{\sigma_t\}$.

The unnormalized a posteriori state $\sigma_t$ can be computed by using the spectral representation \eqref{WSR} of the operators $Q_k$ and its evolution can be obtained by passing through the SSE.

\subsubsection{The stochastic Schr\"odinger equation}\label{SSE}

In order to compute the a posteriori state of our instrument $\Ical_t$ \eqref{inst}, it is convenient to pass through two fictitious  instruments: $\hat {\Jcal_t}$,  associated to a complete set of compatible observables in $\Gamma$, and $\Jcal_t$,  associated to a complete set of compatible observables in $\Gamma[(0,t)]$. This latter instrument has the simple a posteriori state \eqref{sigmaJ},  whose evolution is given by the SSE \eqref{lSSE}.

First of all, let us  imagine, in the Heisenberg picture, that in the time interval $[0,t]$ we measure all the quadratures $Q_k^{\mathrm{out}}(s)=U_s^\dagger Q_k(s)U_s$, $k=1,\ldots,d$, $s\in[0,t]$, and moreover we conclude the measure by observing at time $t$  also the field observables $\hat Q_k(u;t)=U_t^\dagger Q_k(u) U_t$, $k=1,\ldots,d$ and $u<0$ or $u>t$. This is a family of commuting observables, thanks to \eqref{2quadrature} and to \eqref{T>t}, that implies $Q_k^{\mathrm{out}}(s)=U_t^\dagger Q_k(s)U_t$ for every $k$. Then, the instrument $\hat{\Jcal_t}$ associated to this fictitious measurement is given by an expression analogous to \eqref{inst}. Using again the joint pvm $J^{-1}1_AJ$, $A \in \Fscr$, of the quadratures $Q_k(s)$, if the system initial state is pure, $\rho=|r\rangle \langle r|$, $r\in\Hscr$, $\norm{r}=1$, then $\forall F\in \Fscr$,
\begin{equation}\label{INSThat}
\hat{\Jcal_t}(F)[|r\rangle\langle r|]= \Tr_\Gamma\left\{\bigl( \bbone_\Hscr\otimes J^{-1}1_FJ\bigr)|\Psi_t\rangle\langle\Psi_t|\right\},
\end{equation}
where
\[
\Psi_t= U_t\bigl(r\otimes e(f)\bigr).
\]

The isomorphism $J^{-1}$ does not involve the space $\Hscr$ and it can be cycled after $\langle\Psi_t|$; in this way we get
\begin{equation}\label{INST}
\hat{\Jcal_t}(F)[|r\rangle\langle r|]= \int_F |\varphi_t(\omega)\rangle\langle \varphi_t(\omega)|\Qbb(\rmd\omega),
\end{equation}
where $\varphi_t$ is the random $\Hscr$-vector
\begin{equation*}
\varphi_t=J\Psi_t=J\,U_t\bigl(r\otimes e(f)\Bigr).
\end{equation*}
By comparing Eq.\ \eqref{INST} with Eq.\ \eqref{nnap}, we get that the unnormalized a posteriori state $\sigma_t^{\hat \Jcal}$ associated to the instrument $\hat\Jcal$ and to the pre-measurement system state $\rho=|r\rangle\langle r|$ is
\begin{equation*}
\sigma_t^{\hat \Jcal}(\omega)=|\varphi_t(\omega)\rangle\langle\varphi_t(\omega)|.
\end{equation*}
By construction, $\sigma_t^{\hat \Jcal}$ is a random positive trace-class operator, which is $\Fscr$-measurable.

Suppose now that we measure only the quadratures $Q_k^{\mathrm{out}}(s)$, $0\leq s\leq t$, $k=1,\ldots,d$; note that this set of compatible observables is complete in $\Gamma[(0,t)]$, not in $\Gamma[\Rbb]$.
With respect to the previous case, we simply have to drop some commuting observables and thus the new instrument $\Jcal_t$ is just the restriction of $\hat{\Jcal_t}$ to the $\sigma$-algebra $\Fscr_t\subset\Fscr$ and, therefore, the new a posteriori state $\sigma_t^\Jcal$ is the conditional expectation
\begin{equation*}
\sigma_t^\Jcal=\Ebb_\Qbb\bigl[\sigma_t^{\hat\Jcal}\big|\Fscr_t\bigr],
\end{equation*}
which is an $\Fscr_t$-measurable random positive operator. Thanks to the properties of the Hudson-Parthasarathy equation and to the choice of the observed quadratures (local in $(0,t)$ and no $k$ neglected) the a posteriori state is still almost surely pure:
\begin{equation}\label{sigmaJ}
\sigma_t^\Jcal(\omega)=|\phi_t(\omega)\rangle\langle\phi_t(\omega)|, \quad \phi_t=J\,U_t\bigl(r\otimes e(f1_{(0,t)})\bigr).
\end{equation}

\paragraph{The linear SSE}
It is the $(\Fscr_t)$-adapted stochastic process $\phi_t$ that satisfies the linear SSE and we show now how to get it.

By  introducing the Weyl operators $\Wcal_t:=\Wcal(f1_{(0,t)})$ we can write $e(f1_{(0,t)})=\Wcal_t e(0)$ and
\begin{equation}\label{JUW}
\phi_t =J\,U_t\,\Wcal_t\bigl(r\otimes e(0)\bigr).
\end{equation}

By using the definition of the Weyl operators and the quantum  stochastic calculus it is easy to check that $\Wcal_t$ satisfies a quantum SSE \eqref{HPequ} with $R_k=f_k(t)$ and $H_0=0$. Then, the quantum stochastic differential of $U_t\Wcal_t$ is given by
$
\rmd\left(U_t\Wcal_t\right)=\bigl(\left(\rmd U_t\right)\Wcal_t+U_t\,\rmd\Wcal_t+\left(\rmd U_t\right)\rmd\Wcal_t\bigr)$.
This can be computed by using the quantum It\^o table \eqref{Itable} and  exploiting that operators localized in disjoint time intervals commute and that the differentials $\rmd  B_k^{(\dagger)}(t)$ are localized in $(t,t+\rmd t)$ with respect to the factorization \eqref{FockFactorization} of the Fock space; in particular this means that $\rmd B_k(t)$ and $U_t$ commute and the same holds for  $\rmd B_k(t)$ and $\Wcal_t$. The result is
\begin{widetext}
\begin{equation*}
\rmd\left(U_t\Wcal_t\right)=\biggl\{K\rmd t +
\sum_k \left[ \bigl(R_k +f_k(t)\bigr)\rmd B_k^\dagger(t)-\left(R_k^\dagger+\overline{f_k(t)}\right)\rmd B_k(t)
-\left({\textstyle\frac{1}{2}}\abs{f_k(t)}^2+ f_k(t)R_k^\dagger\right) \rmd t\right]\biggr\}U_t\,\Wcal_t.
\end{equation*}
Now, we use this result to compute the differential of $\phi_t$ \eqref{JUW}. The key point is that  \[
\bigl(\rmd B_k(t)\bigr)U_t\,\Wcal_t\left(r\otimes e(0)\right)=U_t\,\Wcal_t\,\rmd B_k(t)\left(r\otimes e(0)\right)=0,\]
so that we can change the coefficient of $\rmd B_k(t)$ as we wish.
By this and the fact that $\abs{\ell_k(t)}=1$, we can write
\begin{equation*}
\rmd\phi_t=J\biggl\{K\rmd t +
\sum_k \left[ \bigl(R_k +f_k(t)\bigr)\overline{\ell_k(t)}\,\rmd Q_k(t)
-\left({\textstyle\frac{1}{2}}\abs{f_k(t)}^2+ f_k(t)R_k^\dagger\right) \rmd t\right]\biggr\}U_t\,\Wcal_t\bigl(r\otimes e(0)\bigr).
\end{equation*}
Finally, by Eqs.\ \eqref{K+H}, \eqref{WSR}, and \eqref{JUW} we obtain the final form \eqref{lSSE} of the linear SSE. We collect this result and the main property of its solution in a theorem.
\begin{theorem}[linear stochastic Schr\"odinger equation]\label{th2}
In the hypotheses of Theorem \ref{th1} and Remark \ref{rem:Wiener}, the random vector $\phi_t=J\,U_t\bigl(r\otimes e(f1_{(0,t)})\bigr)$, where $J$ is the isomorphism \eqref{isoJ} and $f$ is a function such that $f1_{(0,t)}\in L^2(\mathbb{R};\Cbb^d)$, $\forall t>0$, is the unique solution of the It\^o-type stochastic differential equation
\begin{equation}\label{lSSE}
\rmd\phi_t=\biggl\{\sum_k \biggl[\overline{\ell_k(t)}\bigl(R_k+f_k(t)\bigr)\rmd W_k(t)-\frac{1}{2}\left(R_k^\dagger+\overline{f_k(t)}\right)\bigl(R_k+f_k(t)\bigr)\rmd t\biggr]
-\rmi \biggl[H_0+\frac{\rmi}{2}\sum_k\left(\overline{ f_k(t)}\,R_k-f_k(t)R_k^\dagger\right)\biggr]\rmd t\biggr\}
\phi_t;
\end{equation}
the function $\ell $ appears in the definition of $J$ and it is introduced in Eq.\ \eqref{f_ell}. Moreover, $\norm{\phi_t}^2_\Hscr$, $t\geq 0$,  is a $\Qbb$-martingale.
\end{theorem}
\end{widetext}

The fact that $\norm{\phi_t}^2_\Hscr$ is a $\Qbb$-martingale, i.e.
\[
\Ebb_\Qbb\bigl[\norm{\phi_t}^2_\Hscr\big|\Fscr_s\bigr]=\norm{\phi_s}^2_\Hscr, \qquad t\geq s\geq 0,
\]
is proved by computing its stochastic differential.

By the results above, we have that the a posteriori evolution of system $S$, under the continuous measurement $\Jcal_t$, is a stochastic evolution mapping pure states into pure states. In particular, the evolution of the unnormalized pure state $\phi_t$  is given by the linear SSE \eqref{lSSE} and it is Markovian. Such an evolution depends on the interaction \eqref{HPequ} between $S$ and $\Gamma$, on the field initial state $\varrho_\Gamma(f)$ and on the observed quadratures $Q_k^{\mathrm{out}}(s)$.

It is now possible to show that $\norm{\phi_T(\omega)}^2_\Hscr \Qbb(\rmd \omega)$ defines a new probability on $(\Omega, \Fscr_T)$ and that, under this new probability, the process $\psi_t:= \phi_t /\norm{\phi_t}_\Hscr$, $t\in[0,T]$, satisfies a nonlinear stochastic differential equation. This last equation is the nonlinear SSE, which is the starting point for useful numerical simulation. A key point in the change of probability is the fact that $\norm{\phi_t}^2_\Hscr$ is a $\Qbb$-martingale and that the so called Girsanov transformation can be invoked. For the theory of the linear and nonlinear SSE we refer to \cite[Sec.\ 2]{BarG09}.

\subsubsection{The stochastic master equation}
By It\^o calculus, from the SSE \eqref{lSSE} we get the stochastic equation satisfied by the random operator $\sigma_t^\Jcal$ \eqref{sigmaJ}:
\begin{align}\notag
\rmd \sigma_t^\Jcal= &\Lcal(t)[\sigma_t^\Jcal]\, \rmd t
+ \sum_k \Bigl\{ \overline{\ell_k(t)} \bigl( R_k +f_k(t)\bigr)\sigma_t^\Jcal\\{}&+ \sigma_t^\Jcal\, \ell_k(t) \left(R_k^\dagger +\overline{f_k(t)} \right) \Bigr\} \rmd W_k(t), \label{SME1}
\end{align}
where $\Lcal(t)$ is the Liouville operator \eqref{Lop}.

We can now get rid of the hypothesis of a pure initial state and prove that Eq.\ \eqref{SME1} gives the a posteriori evolution for a generic system initial state
\[
\rho=\sum_\ell p_\ell |r_\ell\rangle\langle r_\ell|, \quad \norm{r_\ell}_\Hscr=1, \quad p_\ell >0, \quad \sum_\ell p_\ell=1.
\]
Indeed, if we set
\[
\phi_t^\ell =J\,U_t\,\Wcal_t\bigl(r_\ell\otimes e(0)\bigr),
\]
then, by linearity the process
\begin{equation*}
\sigma_t^{\Jcal}(\omega)=\sum_\ell p_\ell|\varphi_t^\ell(\omega)\rangle\langle\varphi_t^\ell(\omega)|
\end{equation*}
is adapted, satisfies Eq.\  \eqref{SME1}, and gives the a posteriori state of $\Jcal_t$,
\begin{equation}
\Jcal_t(F)[\rho]=\int_F \sigma_t^\Jcal(\omega)\,\Qbb(\rmd \omega), \qquad \forall F\in \Fscr_t.
\end{equation}

Finally, we consider our instrument $\Ical_t$, which is the restriction of $\Jcal_t$ to $\Gscr_t$, so that Eq.\ \eqref{nnap}, defining the unnormalized a posteriori states, holds with
\begin{equation}\label{y}
\sigma_t=\Ebb_\Qbb\bigl[\sigma_t^\Jcal\big|\Gscr_t\bigr]
.
\end{equation}
Let us  recal that $\ell_1(t)=\rme^{-\rmi \vartheta }h(t)$, that the $\Qbb$-mean of any $W_k(s)$ is zero, that $\Gscr_t$ is generated by $W_1$ and that the other components of the Wiener process are independent from the first one. Then, by applying the conditional expectation with respect to $\Gscr_t$ to \eqref{SME1}, we obtain the linear SME for the unnormalized a posteriori states of $\Ical_t$.
\begin{theorem}[Lin.\ stochastic master equation]\label{th2b}
In the hypotheses of Theorem \ref{th1} and Remark \ref{rem:Wiener}, the unnormalized a posteriori states $\sigma_t$ of the instrument $\Ical_t$ \eqref{inst}, realized by the observation of the field quadratures $Q^{\mathrm{out}}(s;\vartheta,h)$, $0\leq s\leq t$, \eqref{Qout}, is the unique solution of the It\^o-type stochastic differential equation
\begin{align}\notag
\rmd \sigma_t=& \Lcal(t)[\sigma_t]\, \rmd t
+ \Bigl\{ \rme^{\rmi \vartheta}\overline{h(t)} \bigl( R_1 +f_1(t)\bigr)\sigma_t\\{}&+ \sigma_t\, \rme^{-\rmi \vartheta}h(t) \left(R_1^\dagger +\overline{f_1(t)} \right) \Bigr\} \rmd W_1(t). \label{SME2}
\end{align}
Moreover, $\Tr\left\{\sigma_t\right\}$, $t\geq 0$,  is a $\Qbb$-martingale.
\end{theorem}

As already said at the end of Sec.\ \ref{sec:inst}, the quantity $\Tr\left\{\sigma_t(\omega)\right\}$ is the density of the physical probability with respect to $\Qbb$. Indeed, from Eqs.\ \eqref{nnap} and \eqref{P+inst} we get
\begin{equation}\label{xx}
\Pbb_{\rho,t}^{\vartheta,h}(G)=\int_G \Tr\left\{\sigma_t(\omega)\right\}\Qbb(\rmd \omega), \qquad \forall G \in \Gscr_t.
\end{equation}
In the stochastic formulation, it is just the martingale property that implies the consistency property \eqref{ts_cons}, which we already encountered in the Fock space formulation.

As already seen, if we define $\rho_t=\sigma_t/\Tr\{\sigma_t\}$, we get the a posteriori state for the instrument $\Ical_t$ and the pre-measurement state $\rho$. By \eqref{etaI}, the a posteriori states are related to the system reduced state by
\begin{equation}\label{zz}
\int_\Omega \rho_t(\omega)\, \Pbb_{\rho,t}^{\vartheta,h}(\rmd \omega) =\int_\Omega \sigma_t(\omega)\, \Qbb(\rmd \omega)=\eta_t.
\end{equation}

It is also possible to prove that, under the physical probability, $\rho_t$ satisfies a nonlinear SME. Moreover, by studying the stochastic differential of  $\Tr\left\{\sigma_t\right\}$ and by using Girsanov theorem, it is possible to prove the following result.

\begin{proposition}[The output and the noise]
Under the physical probability $\Pbb^{\vartheta,h}_{\rho,T}$ the process
\[
\widehat W_1(t):= W_1(t) -2\RE\int_0^t \rme^{\rmi \vartheta} h(s) \Tr\left\{ \bigl(R_1+f_1(s)\bigr) \rho_s\right\} \rmd s
\]
is a standard Wiener process for $t\in[0,T]$.
\end{proposition}

In other terms we can say that, under the physical probability, the instantaneous output $I(t) = \dot W_1(t)$ is the sum of a white noise $\rmd \widehat W _1(t)/\rmd t$ plus a regular signal $2\RE \rme^{\rmi \vartheta} h(t)\Tr\left\{ \bigl(R_1+f_1(t)\bigr) \rho_t\right\} $. Indeed the continuous measurement provides information also on the system $S$ and the output $I(t)$ is interpreted as an imprecise measure at time $t$ of the system observable $\rme^{\rmi \vartheta} h(t)\bigl(R_1+f_1(t)\bigr)+ \mathrm{h.c.}$. In general, white noise and signal turn out to be correlated. Let us stress that this result on the structure of the output is a byproduct of the stochastic representation of the continuous measurements. From this representation and Eq.\  \eqref{zz} the mean value of the observed quadrature at time $t$ is
\begin{align}\notag
\Tr &\left\{\bigl(\bbone_\Hscr\otimes Q(t;\vartheta,h)\bigr)\Sigma_f(t)\right\}=\int_\Omega W_1(t;\omega)\, \Pbb_{\rho,t}^{\vartheta,h}(\rmd \omega)\\ {}&=2\RE\int_0^t \rme^{\rmi \vartheta} h(s)\Tr\left\{ \bigl(R_1+f_1(s)\bigr) \eta_s\right\} \rmd s.\label{mean}
\end{align}

The full theory of the linear and nonlinear SMEs and their relations with the physical probability and the reference probability $\Qbb$ are presented in \cite[Secs.\ 3 and 5]{BarG09}.

\subsection{Quantum filtering}

Up to now we have described a continuous measurement of $S$ in its Schr\"odinger picture: first we have coupled $S$ with a bosonic field by the quantum SSE, then we have chosen some compatible field observable, finally we have deduced the evolution for the (unnormalized) a posteriori state $\sigma_t$ and for the observed output $I(t)$, or $W(t)$.

The same procedure in Heisenberg picture leads to the \emph{quantum filtering theory} \cite{Bel88,BouVHJ07,BelE08,BouH08,Bou08}. Let us quickly connect the two theories just for people that already know quantum filtering.

Let $\Bscr(\Hscr)$ denote the von Neumann algebra of the bounded operators on $\Hscr$. Then, in Heisenberg picture, the evolution of the system observables is given by the \emph{quantum stochastic flow} $j_t:\Bscr(\Hscr)\to\Bscr(\Hscr\otimes\Gamma)$,
\[
X\mapsto j_t(X)=U_t^\dagger(X\otimes\bbone_\Gamma)U_t.
\]
By using \eqref{Qout}, we get immediately $[j_t(X),Q^{\mathrm{out}}(s;\vartheta,h)]=0$ for all $0\leq s\leq t$ and for all $X\in \Bscr(\Hscr)$; this is the \emph{nondemolition property} of the quantum filtering theory, while the commutatation of the observed quadratures \eqref{snd} is called \emph{self-nondemolition property}. Thus, the measurement of an arbitrary system observable at time $t$ is compatible with our continuous measurement of the output quadratures in the time interval $[0,t]$. Therefore, we can consider $\mathscr{Q}_t=\operatorname{vN}(Q^{\mathrm{out}}(s;\vartheta,h)|0\leq s\leq t)$, the commutative von Neumann algebra generated by the observed quadratures in $[0,t]$ and  we can introduce the main object of quantum filtering: the \emph{conditional expectation} of $j_t(X)$ with respect to $\mathscr{Q}_t$
\[
\mu_t(X)=\Ebb[j_t(X)|\mathscr{Q}_t]\in\mathscr{Q}_t,
\]
defined by, $\forall Y\in\mathscr{Q}_t$,
\[
\Tr\{\mu_t(X)Y\left(\rho\otimes\varrho_\Gamma(f)\right)\}=\Tr\{j_t(X)Y\left(\rho\otimes\varrho_\Gamma(f)\right)\}.
\]
It is the observable giving the value to be expected for $j_t(X)$ at time $t$, on the basis of the results already observed for the field quadratures, in the case the global initial state were $\rho\otimes\varrho_\Gamma(f)$ and in the case also $j_t(X)$ were actually observed at time $t$. Note that it depends on the global evolution ($U_t$), on the initial states of both $S$ ($\rho$) and the field (through $f$) and, of course, on the observed quadratures (through $\vartheta$ and $h$).

The conditional expectation $\mu_t(X)$ is strictly related to the a posteriori states, as it is possible to prove that
\[
\mu_t(X)=U_t^\dagger\int_\Omega\Tr\{X\rho_t(\omega)\}E_\vartheta^h(\rmd\omega)U_t=\frac{\varpi_t(X)}{\varpi_t(\bbone_\Hscr)},
\]
where $E_\vartheta^h$ is the pvm introduced in Sec.\ \ref{QtoC} and
\[
\varpi_t(X)=U_t^\dagger\int_\Omega\Tr\{X\sigma_t(\omega)\}E_\vartheta^h(\rmd\omega)U_t
\]
is the unnormalized conditional expectation. Then the SMEs satisfied by the a posteriori states $\rho_t$ and $\sigma_t$ can be translated into the \emph{quantum filtering equations} satisfied by the conditional expectations $\mu_t(X)$ and $\varpi_t(X)$. These are quantum stochastic differential equations, respectively nonlinear and linear, which are driven by the commutative observables $Q^{\mathrm{out}}(t;\vartheta,h)$, so that, actually, they are classical stochastic differential equations by the spectral theorem. For example, the linear quantum filtering equation (the analog of the classical Zakai equation) is
\begin{multline*}
\rmd\varpi_t(X)=\varpi_t\big(\Lcal^*(t)[X]\big)\, \rmd t + \varpi_t\Big(X\rme^{\rmi \vartheta}\overline{h(t)} \bigl( R_1 +f_1(t)\bigr) \\
+ \rme^{-\rmi \vartheta}h(t) \left(R_1^\dagger +\overline{f_1(t)} \right) X\Big)\rmd Q^{\mathrm{out}}(t;\vartheta,h);
\end{multline*}
$\Lcal^*(t)$ is the adjoint of the Liouville operator.
Note that this equation depends on the global evolution ($U_t$), on the field initial state ($f$), on the observed quadratures ($\vartheta,h$), but not on the system state $\rho$, which instead gives the initial condition $\varpi_0(X)=\Tr\{X\rho\}\bbone_\Gamma$.

Another important concept of quantum filtering theory is the \emph{innovation process} $Z(t)$ of the observation, defined by
\[
Z(t)=Q^{\mathrm{out}}(t;\vartheta,h)-\int_0^t\mu_s\left(\rme^{\rmi \vartheta} h(s)\bigl(R_1+f_1(s)\bigr)+ \mathrm{h.c.}\right) \rmd s,\]
which is distributed as a Wiener process under the initial state $\rho\otimes\varrho_\Gamma(f)$. This is nothing but  our process $\widehat W_1(t)$, view as a family of operators on $\Hscr\otimes\Gamma$ and put in Heisenberg picture, i.e. $Z(t)= U_t^\dagger J^{-1}\widehat W_1(t) JU_t$.

\subsection{Characteristic functional and moments}\label{co,pm}

In the study of stochastic processes it is often useful to have explicit formulae for the moments, for instance for the second order moments, which determine the spectrum of the process; see Sec.\ \ref{sec:def_sp}. In the case of our output, the mean function is given by Eq.\ \eqref{mean}; to get the higher-order moments it is useful to introduce the \emph{characteristic functional}, which is the functional Fourier transform of the probability distribution of the process.

Let us denote by $\Ebb_{\rho,t}^{\vartheta,h}$ the expectation with respect to the physical probability
$\Pbb_{\rho,t}^{\vartheta,h}$.
By recalling that the output is represented by $W_1$, the characteristic functional up to time $t>0$ is
\begin{equation}\label{ch_fu}
\Phi_t(k;\vartheta,h) = \Ebb_{\rho,t}^{\vartheta,h}\bigg[\exp \biggl\{ \rmi  \int_0^t k(s)\, \rmd W_1(s)
\biggr\}\bigg];
\end{equation}
the argument $k$ is any real test function in $L^\infty(\Rbb_+)$.

By functional differentiation with respect to the test function one gets all the moments of the process, as done in Sec.\ \ref{subsec:om}. Moreover, when the characteristic functional is given, one gets the probabilities by anti-Fourier transform. For instance, the finite-dimensional probability densities of the increments $W_1(t_1)-W_1(t_0)$, $W_1(t_2)-W_1(t_1)$, \ldots, $W_1(t_n)-W_1(t_{n-1})$, with $0\leq t_0<t_1<\cdots   < t_n \leq t$, are given by
\[
\frac 1 {(2\pi)^n}\int_{\Rbb^{n}}\rmd \kappa_1\cdots \rmd \kappa_n  \biggl(\prod_{j=1}^n  \rme^{-\rmi\kappa_j \cdot
x_j}\biggr)\Phi_t(k;\vartheta,h) ,
\]
where we have  introduced the test function $k(s)= \sum_{j=1}^n 1_{(t_{j-1},t_j)}(s)
\,\kappa_j$.

\subsubsection{Characteristic operators}

As our measurement is an indirect observation of $S$ performed by a direct observation of $\Gamma$, the characteristic functional \eqref{ch_fu} can be expressed either in terms of the system $S$ only or in terms of the fields only. First, we want to express it in terms of the quantum observables \eqref{quadrature}.
We introduce the \emph{characteristic operator} $\widehat \Phi_t(k;\vartheta,h)$, the
Fourier transform of the pvm  $E_\vartheta^h$ \cite[Sec.\ 11.4.2]{GarZ00}, \cite[Sec.\ 3.2]{Bar06}:
\begin{equation}
\widehat \Phi_t(k;\vartheta,h) = \int_\Omega \exp \biggl\{ \rmi  \biggl(\int_0^t k(s)\, \rmd W_1(s)\biggr)(\omega)
\biggr\} E_\vartheta^h(\rmd \omega).\label{Phiincrements}
\end{equation}

By using the representation \eqref{pvm} of the pvm, the correspondence between $Q$ and $W_1$ \eqref{WSR}, and the definition of $Q$ \eqref{quadrature}, we get
\begin{align}\notag
\widehat \Phi_t(k;\vartheta,h) &= J^{-1}\exp \biggl\{ \rmi  \int_0^t k(s)\, \rmd W_1(s)\biggr\}J
\\  \notag
{}&=
\exp \biggl\{ \rmi  \int_0^t k(s)\, \rmd Q(s;\vartheta,h)
\biggr\}
\\ {}&=
\exp \biggl\{ \rmi  \rme^{-\rmi \vartheta}\int_0^t k(s) h(s)\, \rmd
B_1^\dagger(s) - \textrm{h.c.}
\biggr\}.\label{WQh}
\end{align}
By comparing this expression with Eq.\  \eqref{Wcal}, we see that the characteristic operator is the unitary Weyl operator
\[
\widehat \Phi_t(k;\vartheta,h)=\Wcal(q), \qquad  q_j(s) = \delta_{j1}\rmi \rme^{-\rmi \vartheta} k(s) h(s)1_{[0,t]}(s).
\]

Then, by using the expression \eqref{P(G)} of the physical probability, the characteristic functional can be written as
\begin{align}\notag
\Phi_t(k;\vartheta,h) &= \Tr \left\{\widehat\Phi_t(k;\vartheta,h)\Sigma_f(t)\right\}
\\ \notag {}= \Tr &\left\{\exp \biggl[ \rmi  \int_0^t k(s)\, \rmd
Q^{\mathrm{out}}(s;\vartheta,h) \biggr] \rho\otimes \varrho_\Gamma(f) \right\}
\\ {}&
 =\Tr_\Gamma \left\{\widehat\Phi_t(k;\vartheta,h)\Pi_f(t)\right\}.\label{characteristicf}
\end{align}
The last step is due to the fact that $\widehat\Phi_t(k;\vartheta,h)$ depends only on field operators; recall that $\Sigma_f(t)$ is the state of the total system $S$ plus fields, while $\Pi_f(t)$ is the reduced state of the fields \eqref{redst}.

Finally, let us define the  \emph{reduced characteristic operator} $\Gcal_t$ as the functional Fourier transform of the instrument \eqref{inst} \cite{BarLP82,BarL85,Bar06,BarG09}:
\begin{equation}\label{Gcal}
\Gcal_t(k;\vartheta,h) = \int_\Omega \exp \biggl\{ \rmi  \biggl(\int_0^t k(s)\, \rmd W_1(s)\biggr)(\omega)
\biggr\} \Ical_t(\rmd \omega).
\end{equation}
It can be shown that $\Gcal_t$ satisfies a closed differential equation, a kind of modification of the master equation \cite{BarL85}.
Then, by the representation \eqref{P+inst} of the physical probabilities, we get a further  expression of the characteristic functional: \begin{equation*}
\Phi_t(k;\vartheta,h) = \Tr_\Hscr\left\{\Gcal_t(k;\vartheta,h)[\rho]\right\}.
\end{equation*}

\subsubsection{The output moments}\label{subsec:om}
By functional differentiation of the characteristic functional, we get all the moments of the classical output
process. Let us introduce the formal time derivatives $I(t)=\dot W_1(t)$ and $\hat I(t)=\dot Q(t;\vartheta,h)$; from \eqref{WQh} and \eqref{characteristicf} we obtain immediately the expressions of mean function and autocorrelation function:
\begin{subequations}\label{moments}
\begin{align}\notag
\Ebb_{\rho,T}^{\vartheta,h}[I(t)]&= \Tr \left\{\dot Q(t;\vartheta,h)\Sigma_f(T)\right\}= \Tr_\Gamma \left\{\hat I(t)\Pi_f(T)\right\} \\
&=2\RE\left( \rme^{\rmi \vartheta}\, \overline{h(t)}\Tr_\Gamma \left\{ b_1(t)\,\Pi_f(T)\right\} \right),
\end{align}
\begin{align}\notag
\Ebb_{\rho,T}^{\vartheta,h}&[I(t)I(s)] = \Tr_\Gamma \left\{\hat I(t)\hat I(s)\Pi_f(T)\right\} \\  \notag
{}&=\delta(t-s)
+ 2\RE\Bigl(\overline{h(s)}\Tr_\Gamma\Bigl\{\Bigl( h(t)b_1^\dagger(t)\\ {}&+\rme^{2\rmi \vartheta} \, \overline{h(t)}\,b_1(t) \Bigr)b_1(s)\,\Pi_f(T)\Bigr\}\Bigr), \label{moment2}
\end{align}
\end{subequations}
where $T>t$, $T>s$. Analogous formulae hold for higher-order moments. Let us note that the order of the operators $\hat I(t)$ and $\hat I(s)$ in \eqref{moment2} does not matter, because they commute. Moreover, the moments of the classical process $I(t)$ are expressed in terms of quantum means and quantum correlations of the fields \cite[p.\ 165 and Sec.\ 11.3.2]{GarZ00}: $\Tr_\Gamma \left\{ b_1(t)\,\Pi_f(T)\right\} $,
$\Tr_\Gamma \left\{b_1^\dagger(t) b_1(s)\,\Pi_f(T)\right\} $, $\Tr_\Gamma \left\{b_1(t) b_1(s)\,\Pi_f(T)\right\} $,
and the complex conjugated expressions. The fields are all in normal order because we put in evidence the delta term coming out from a commutator.

By studying the properties of the reduced characteristic operator \eqref{Gcal}, it is possible to prove that all the moments of our classical output can be expressed by means of quantities concerning only system $S$ \cite{Bar06}.  For the mean and autocorrelation functions the final result is \cite[Secs.\ 3.3, 3.5]{Bar06}
\begin{subequations}\label{1moment}
\begin{equation}
\Ebb_{\rho,T}^{\vartheta,h}[I(t)]=2\RE\left(\Tr_\Hscr\left\{Z(t)\eta_t\right\}\right),
\end{equation}
\begin{align}\notag
\Ebb_{\rho,T}^{\vartheta,h}[I(t)I(s)] &= \delta(t-s)  +2\RE\bigl(\Tr_\Hscr\bigl\{Z(t_2) \\ {}&\times \Upsilon(t_2,t_1)\left[Z(t_1)\eta_{t_1}+\eta_{t_1}
Z(t_1)^\dagger\right] \bigr\}\bigr),
\end{align}
\end{subequations}
where $t_2=t\vee s$, $t_1=t\wedge s$ and
\begin{equation*}
Z(t) := \rme^{\rmi\vartheta}\, \overline{h(t)}\,\left(R_1+f_1(t)\right).
\end{equation*}
These expressions are more useful for computations, while the expressions \eqref{moments} are better suited for theoretical considerations; cf.\ \cite[Sec.\ 5.4.6]{GarZ00}.

\section{The spectrum of the output}\label{sec:spectrum}
Inside the theory of continuous measurements, the output of the measurement is a classical stochastic process, even if its distribution is determined by quantum mechanics; so, the spectrum of the output can be introduced by using the classical definition of spectrum of a stochastic process \cite{BarGL08,BarG08,BarGL09}.

Let us stress the importance of this approach. As we describe the output by a classical stochastic process, we can define its spectrum by a classical definition, which is clearly related to the measurement performed in the real lab. Nevertheless, since the spectrum depends on the distribution of the output, and since this last is determined by our fully quantum model, we can clearly relate the classical properties of the spectrum with the quantum properties of systems $S$ and $\Gamma$.

\subsection{The spectrum of a stationary process}\label{sec:def_sp}
In the classical theory of stochastic processes, the spectrum is related to the Fourier
transform of the autocorrelation function \cite{How02}. Let $Y$ be a stationary real stochastic process
with finite moments; then, the mean is independent of time $ \Ebb[Y(t)]=\Ebb[Y(0)]=:m_Y$, $\forall
t\in \Rbb$, and the second moment is invariant under time translations: $ \forall t,s\in \Rbb$,
\begin{equation}\label{2_st}
\Ebb[Y(t)Y(s)]=\Ebb[Y(t-s)Y(0)]=:R_Y(t-s).
\end{equation}
The function $R_Y(t)$, $t\in \Rbb$, is called the \emph{autocorrelation function} of
the process. Obviously, we have $ \Cov\left[Y(t),Y(s)\right]=R_Y(t-s)-m_Y^{\,2}$.

The \emph{spectrum} of the stationary stochastic process $Y$ is the Fourier transform of its
autocorrelation function:
\begin{equation}
S_Y(\mu):=\int_{-\infty}^{+\infty}\rme^{\rmi \mu t}R_Y(t)\, \rmd t.
\end{equation}
This formula has to be intended in the sense of distributions. If
$\Cov\left[Y(t),Y(0)\right]\in L^1(\Rbb)$, we can write
\begin{equation}
S_Y(\mu):=2\pi m_Y^{\,2} \delta(\mu)+\int_{-\infty}^{+\infty}\rme^{\rmi \mu
t}\Cov\left[Y(t),Y(0)\right] \rmd t.
\end{equation}
By the properties of the covariance, the function $\Cov\left[Y(t),Y(0)\right]$ is positive
definite and, by the properties of positive definite functions, this implies
$\int_{-\infty}^{+\infty}\rme^{\rmi \mu t}\Cov\left[Y(t),Y(0)\right] \rmd t\geq 0$;
then, also $S_Y(\mu)\geq 0$.

By using the stationarity and some tricks on multiple integrals, one can check that an
alternative expression of the spectrum is
\begin{equation}\label{eq:sp2}
S_Y(\mu)= \lim_{T\to + \infty}\frac 1 T \Ebb\left[\abs{\int_0^T  \rme^{\rmi \mu t}Y(t)\,\rmd t
}^2\right].
\end{equation}
The advantage now is that the positivity of the spectrum appears explicitly and only positive times are involved \cite{How02}.
Expression \eqref{eq:sp2} can be generalized also to processes which are stationary only in
some asymptotic sense and to singular processes as our output current $I(t)$.

\subsection{The spectrum of the output in a finite time horizon}
Let us consider our output $I(t)= \rmd W_1(t)/\rmd t$ under the physical probability
$\Pbb_{\rho,T}^{\vartheta,h}$. We call ``spectrum up to time $T$'' of $I(t)$ the quantity
\begin{equation}\label{eq:spI(t)}
S_T(\mu;\vartheta)= \frac 1 T \Ebb_{\rho,T}^{\vartheta,h}\left[\abs{\int_0^T  \rme^{\rmi \mu
t}\,\rmd W_1(t) }^2\right].
\end{equation}
Note that the spectrum is an even function of $\mu$: $S_T(\mu;\vartheta)=S_T(-\mu;\vartheta)$. When the limit $T\to +\infty$ exists, we can speak of \emph{spectrum of the output,} but this existence depends on the specific properties of the concrete model.

By writing the second moment defining the spectrum as the square of the mean plus the
variance, the spectrum splits in an elastic or coherent part and in an inelastic or incoherent
one:
\begin{subequations}
\begin{equation}
S_T(\mu;\vartheta)=
S_T^{\mathrm{el}}(\mu;\vartheta)+S_T^{\mathrm{inel}}(\mu;\vartheta),
\end{equation}
\begin{equation}
S_T^{\mathrm{el}}(\mu;\vartheta)= \frac 1 T \abs{\Ebb_{\rho,T}^{\vartheta,h}\left[\int_0^T
\rme^{\rmi \mu t}\,\rmd W_1(t)\right] }^2,
\end{equation}
\begin{align}\notag
S_T^{\mathrm{inel}}(\mu;\vartheta)&= \frac 1 T \Var_{\rho,T}^{\vartheta,h}\left[\int_0^T
\cos\mu t\,\rmd W_1(t) \right]\\ {}& + \frac 1 T\Var_{\rho,T}^{\vartheta,h}\left[\int_0^T \sin\mu
t\,\rmd W_1(t) \right] .
\end{align}
\end{subequations}
Let us note that
\begin{equation}
S_T^{\mathrm{el}}(\mu;\vartheta)=S_T^{\mathrm{el}}(-\mu;\vartheta), \quad
S_T^{\mathrm{inel}}(\mu;\vartheta)=S_T^{\mathrm{inel}}(-\mu;\vartheta).
\end{equation}

In Sec.\ \ref{co,pm} we have seen two ways of expressing the output moments, by means of field operators or by means of system operators.

By using the expression \eqref{moment2} for the autocorrelation function of the output, we obtain
\begin{align}\notag
&S_T(\mu;\vartheta)= 1+\frac 2 T\int_0^T \rmd t \int_0^T \rmd s\; \rme^{\rmi \mu
\left(t- s\right)}\,
\RE\Bigl(\overline{h(s)} \\ {}&\times\Tr_\Gamma\left\{\left( h(t)b_1^\dagger(t)+\rme^{2\rmi \vartheta} \, \overline{h(t)}\,b_1(t) \right)b_1(s)\,\Pi_f(T)\right\}\Bigr).\label{b:spI(t)}
\end{align}
Let us note that Eq.\ \eqref{b:spI(t)} expresses the spectrum as a Fourier transform (in a finite time interval) of a normal ordered quantum correlation function of the field; cf.\  \cite[Sec.\ 9.3.2]{Car08}. Let us recall that the function $h(t)$ has modulus one and represents the phase contribution coming from the interference with the wave of the local oscillator.

By using the expressions \eqref{1moment} for the first two moments we get the spectrum in a
form which involves only system operators:
\begin{subequations}\label{spectrumS}
\begin{align}\notag
&S_T^{\mathrm{el}}(\mu;\vartheta)= \frac 1 T \abs{\int_0^T \rme^{\rmi \mu t}\Tr_\Hscr \left\{
\left(Z(t)+Z(t)^\dagger\right) \eta_t\right\}\rmd t}^2
\\
{}&=\frac 4 T \abs{\int_0^T \rme^{\rmi \mu t}\RE\left[\rme^{\rmi \vartheta}\, \overline{h(t)}\bigr(\Tr_\Hscr \left\{
R_1 \eta_t\right\}+f_1(t)\bigr)\right]\rmd t}^2,
\end{align}
\begin{align}\notag
&S_T^{\mathrm{inel}}(\mu;\vartheta) = 1  +\frac 2 T \int_0^T\rmd t  \int_0^t \rmd s
\,\cos\mu(t-s)
\\ {}&\times \Tr_\Hscr\left\{\left(\tilde Z(t)+\tilde Z(t)^\dagger\right)\Upsilon(t,s)\left[\tilde
Z(s)\eta_s+\eta_s \tilde Z(s)^\dagger\right] \right\},
\end{align}
\end{subequations}
where $\Upsilon(t,s)$ is the propagator \eqref{propagator} of the reduced dymamics and
\begin{equation*}
\tilde Z(t) = Z(t) - \Tr_\Hscr\left\{Z(t) \eta_t\right\}=\rme^{\rmi \vartheta}\, \overline{h(t)}\bigr(R_1-\Tr_\Hscr \left\{
R_1 \eta_t\right\}\bigr).
\end{equation*}

\subsection{Properties of the spectrum and the Heisenberg uncertainty relations}

Equations \eqref{spectrumS} give the spectrum in terms of the reduced description of system
$S$ (the fields are traced out); this is useful for concrete computations. But the general
properties of the spectrum are more easily obtained by working with the fields; so, here we
trace out system $S$ and we start from expression \eqref{b:spI(t)}.

Let us define the field operators
\begin{subequations}
\begin{gather}\label{Qmu}
Q_T(\mu;\vartheta)= \frac 1 {\sqrt{T}}\int_0^T\rme^{\rmi \mu t}\,\rmd Q(t;\vartheta,h),
\\
\tilde Q_T(\mu;\vartheta)= Q_T(\mu;\vartheta)-\Tr_\Gamma\left\{\Pi_f(T)Q_T(\mu;\vartheta)
\right\};
\end{gather}
\end{subequations}
the local oscillator wave $h$ is fixed.
Let us stress that $Q_T(\mu;\vartheta)$ commutes with its adjoint and that
$Q_T(\mu;\vartheta)^\dagger=Q_T(-\mu;\vartheta)$. By using Eqs.\ \eqref{characteristicf} and
\eqref{moments} and taking first the trace over $\Hscr$, we get
\begin{subequations}\label{eq:sp2I(t)}
\begin{gather}\label{eq:tot}
S_T(\mu;\vartheta)= \Tr_\Gamma\left\{\Pi_f(T)
Q_T(\mu;\vartheta)^\dagger Q_T(\mu;\vartheta)\right\}\geq 0,
\\ \label{eq:el}
S_T^{\mathrm{el}}(\mu;\vartheta)=  \abs{ \Tr_\Gamma\left\{\Pi_f(T)
Q_T(\mu;\vartheta)\right\}}^2\geq 0,
\\
S_T^{\mathrm{inel}}(\mu;\vartheta)=\Tr_\Gamma\left\{\Pi_f(T) \tilde
Q_T(\mu;\vartheta)^\dagger\tilde Q_T(\mu;\vartheta)\right\}\geq 0.
\end{gather}
\end{subequations}

\subsubsection{Spectrum and field modes}
To elaborate the previous expressions it is useful to introduce annihilation and creation
operators for bosonic temporal modes, as in Sec.\ \ref{TmWo}:
\begin{equation}\label{a(mu)}
a_T(\mu):= \frac 1 {\sqrt{T}} \int_0^T \rme^{\rmi \mu t} \overline{h(t)}\, \rmd B_1(t)\equiv c_1(g_T^\mu),
\end{equation}
\begin{equation}
g_T^\mu(t):=  \frac {\rme^{-\rmi \mu t}} {\sqrt{T}}\,h(t) 1_{[0,T]}(t).
\end{equation}
The operators $a_T(\mu)$, $a_T^\dagger(\mu)$ are true bosonic modes, as they satisfy the CCR
\begin{subequations}\label{muCCR}
\begin{gather}
[a_T(\mu), a_T(\mu^\prime)]= [a_T^\dagger(\mu), a_T^\dagger(\mu^\prime)]=0,
\\
[a_T(\mu), a_T^\dagger(\mu)]=  1 .
\end{gather}
\end{subequations}
However, for finite $T$ these modes are only approximately orthogonal, as we get
\begin{equation}\label{approxcomm}
[a_T(\mu), a_T^\dagger(\mu^\prime)]=
\frac{\rme^{\rmi (\mu -\mu^\prime)T}-1}{ \rmi  (\mu -\mu^\prime)T} \quad \text{for } \
\mu^\prime\neq \mu.
\end{equation}

Then, from Eqs.\ \eqref{quadrature}, \eqref{Qmu}, \eqref{eq:tot} we have easily
\begin{equation}\label{twomodequadrature}
Q_T(\mu;\vartheta)=\rme^{\rmi \vartheta} a_T(\mu)+\rme^{-\rmi \vartheta}
a_T^\dagger(-\mu),
\end{equation}
\begin{align}\notag
S_T(\mu;\vartheta)&= \Tr_\Gamma\Bigl\{\left( \rme^{-\rmi \vartheta}
a_T^\dagger(-\mu) + \rme^{\rmi \vartheta} a_T(\mu) \right)\\ {}&\times \Pi_f(T)\left(
\rme^{-\rmi\vartheta} a_T^\dagger(\mu) +\rme^{\rmi \vartheta} a_T(-\mu) \right) \Bigr\}.
\label{Saadagger}\end{align}
Let us stress that only two field modes contribute to the spectrum for $\mu\neq 0$, and only one mode in the case of $\mu=0$.
By using the CCR \eqref{muCCR} we get the normal ordered version of \eqref{Saadagger}:
\begin{align}\notag
S_T&(\mu;\vartheta)=1 \\ \notag{}&+ \Tr_\Gamma\Bigl\{\Pi_f(T) \Bigl(a_T^\dagger(\mu)a_T(\mu)
+ a_T^\dagger(-\mu)  a_T(-\mu)
\\ {}&+\rme^{-2\rmi \vartheta}
a_T^\dagger(\mu) a_T^\dagger(-\mu) +\rme^{2\rmi \vartheta}
a_T(-\mu) a_T(\mu)\Bigr)\Bigr\}.
\label{Sadaggera}\end{align}
Note that Eq.\ \eqref{approxcomm} played no role in the normal ordering operation.

By Eqs.\ \eqref{eq:el} and \eqref{twomodequadrature} we get for the elastic part of the spectrum
\begin{align}\notag
S_T^{\mathrm{el}}(\mu;\vartheta)&=\Big|\Tr_\Gamma\left\{\Pi_f(T)
a_T(\mu)\right\} \\ {}&+ \rme^{-2\rmi \vartheta}\Tr_\Gamma\left\{\Pi_f(T)
a_T^\dagger(-\mu)\right\}\Big|^2.\label{elS}
\end{align}

To obtain a similar expression also for the inelastic part, it is convenient to introduce the operators
\begin{equation}
\tilde a_T(\mu) := a_T(\mu) - \Tr_\Gamma \left\{ \Pi_f(T)a_T(\mu)\right\},
\end{equation}
which satisfy the same commutation relations \eqref{muCCR}, \eqref{approxcomm} as the operators $ a_T(\mu)$ and their adjoint.
Then, we get
\begin{equation}
\tilde Q_T(\mu;\vartheta)=\rme^{\rmi \vartheta} \tilde a_T(\mu)+\rme^{-\rmi \vartheta}
\tilde a_T^\dagger(-\mu),
\end{equation}
\begin{align}\notag
S_T^{\mathrm{inel}}&(\mu;\vartheta)=1 \\ \notag{}&+ \Tr_\Gamma\Bigl\{\Pi_f(T) \Bigl(\tilde a_T^\dagger(\mu)\tilde a_T(\mu)
+\tilde a_T^\dagger(-\mu) \tilde a_T(-\mu)\\ {}&+\rme^{-2\rmi \vartheta}
\tilde a_T^\dagger(\mu)\tilde a_T^\dagger(-\mu) +\rme^{2\rmi \vartheta}
\tilde a_T(-\mu)\tilde a_T(\mu)\Bigr)\Bigr\}.\label{thetadependence}
\end{align}

\subsubsection{Spectra of complementary quadratures}
Let us consider two choices of the phase $\vartheta$: $\vartheta$ and $\vartheta\pm \pi/2$. From Eq.\ \eqref{2quadrature} we get \[
[Q(t;\vartheta,h),Q(s;\vartheta\pm \pi/2,h)]=\mp 2\rmi \left(t\wedge s\right), \]
which means that we are considering two incompatible field quadratures, measured by two different setups. Indeed, the value of $\vartheta$ can be changed by changing the optical paths of the emited light and local oscillator.
For these quadratures  we have the important bounds and relations given in the following theorem.

\begin{theorem}\label{Th1}
 For every $\vartheta$ and $\mu$ we have the following relations:
\begin{subequations}\label{=hetero}
\begin{align}\notag
&\frac 1 2 \bigl(S_T^{\mathrm{el}}(\mu;\vartheta)+S_T^{\mathrm{el}}(\mu;\vartheta\pm {\textstyle \frac
\pi 2})\bigr) \\ {}&={}\abs{\Tr_\Gamma\left\{\Pi_f(T)
a_T(\mu)\right\}}^2 + \abs{\Tr_\Gamma\left\{\Pi_f(T)
a_T(-\mu)\right\}}^2,
\end{align}
\begin{align}\notag
\frac 1 2 &\bigl(S_T^{\mathrm{inel}}(\mu;\vartheta)+S_T^{\mathrm{inel}}(\mu;\vartheta\pm {\textstyle \frac
\pi 2})\bigr)=1  \\ {}&+ \Tr_\Gamma\Bigl\{\Pi_f(T) \Bigl(\tilde a_T^\dagger(\mu)\tilde a_T(\mu)
+\tilde a_T^\dagger(-\mu)\tilde a_T(-\mu)\Bigl) \Bigr\},\label{thetaindependence}
\end{align}
\end{subequations}
\begin{align}\notag
&\sqrt{S_T^{\mathrm{inel}}(\mu;\vartheta)S_T^{\mathrm{inel}}(\mu;\vartheta\pm {\textstyle
\frac\pi 2})}\geq 1
\\  \label{semibound}
{}&+ \abs{\Tr_\Gamma\left\{\Pi_f(T) \left( \tilde a_T^\dagger(\mu)\tilde a_T(\mu)
-\tilde a_T^\dagger(-\mu)\tilde a_T(-\mu)\right)\right\}}.
\end{align}
Then, independently of the system state $\rho$, of the field state $\varrho_\Gamma(f)$, of the function $h$ and of the
Hudson-Parthasarathy evolution $U$, the following bounds hold:
\begin{equation}\label{Heisenberg}
S_T^{\mathrm{inel}}(\mu;\vartheta)S_T^{\mathrm{inel}}(\mu;\vartheta\pm {\textstyle \frac
\pi 2}) \geq 1,
\end{equation}
\begin{equation}\label{arithbound}
\frac 1 2 \left(S_T^{\mathrm{inel}}(\mu;\vartheta) + S_T^{\mathrm{inel}}(\mu;\vartheta\pm
{\textstyle \frac\pi 2})\right) \geq 1.
\end{equation}
\end{theorem}

\begin{proof}
First of all from Eqs.\ \eqref{elS}, \eqref{thetadependence} we get Eqs.\ \eqref{=hetero}.
Then, the bound \eqref{arithbound} comes immediately from Eq.\ \eqref{thetaindependence}.

The bound \eqref{Heisenberg} is a trivial consequence of Eq.\ \eqref{semibound}.

To prove the bound \eqref{semibound}, we write
\begin{align*}
S_T^{\mathrm{inel}}(\mu;\vartheta)&= \Tr_\Gamma\Bigl\{\left( \rme^{-\rmi \vartheta}
\tilde a_T^\dagger(-\mu) + \rme^{\rmi \vartheta} \tilde a_T(\mu) \right)\\ {}&\times \Pi_f(T)\left(
\rme^{-\rmi\vartheta} \tilde a_T^\dagger(\mu) +\rme^{\rmi \vartheta} \tilde a_T(-\mu) \right) \Bigr\}.
\end{align*}

The usual tricks to derive the Heisenberg-Scr\"odinger-Robertson uncertainty relations can be
generalized also to non-selfadjoint operators \cite{Hol01}. For any choice of the state
$\varrho$ and of the operators $X_1$, $X_2$ (with finite second moments with respect to
$\varrho$) the $2\times 2$ matrix with elements $ \Tr\left\{X_i\varrho X_j^\dagger\right\}$ is
positive definite and, in particular, its determinant is not negative. Then, we have
\begin{align*}
\Tr&\left\{X_1\varrho X_1^\dagger\right\}\Tr\left\{X_2\varrho X_2^\dagger\right\}\geq
\abs{\Tr\left\{X_1\varrho X_2^\dagger\right\}}^2 \\ {}&\geq \abs{\IM\Tr\left\{X_1\varrho
X_2^\dagger\right\}}^2= \frac 1 4 \abs{\Tr\left\{\varrho\left( X_2^\dagger X_1-X_1^\dagger X_2\right)\right\}}^2.
\end{align*}
By taking  $\varrho =\Pi_f(T)$, \[
X_1=\rme^{-\rmi \vartheta} \tilde a_T^\dagger(-\mu)
+\rme^{\rmi \vartheta} \tilde a_T(\mu), \]
\[
X_2=\mp \rmi \rme^{-\rmi \vartheta}
\tilde a_T^\dagger(-\mu) \pm\rmi\rme^{\rmi \vartheta} \tilde a_T(\mu), \]
we get
\begin{align*}
&S_T^{\mathrm{inel}}(\mu;\vartheta)S_T^{\mathrm{inel}}(\mu;\vartheta\pm {\textstyle \frac
\pi 2})
\\ {}&\geq \abs{1 + \Tr_\Gamma\left\{\Pi_f(T) \left( \tilde a_T^\dagger(-\mu)\tilde a_T(-\mu)
-\tilde a_T^\dagger(\mu)\tilde a_T(\mu)\right)\right\}}^2.
\end{align*}
But we can change $\mu$ in $-\mu$ and we have also
\begin{align*}
&S_T^{\mathrm{inel}}(\mu;\vartheta)S_T^{\mathrm{inel}}(\mu;\vartheta\pm {\textstyle
\frac\pi 2}) =S_T^{\mathrm{inel}}(-\mu;\vartheta)S_T^{\mathrm{inel}}(-\mu;\vartheta\pm
{\textstyle \frac\pi 2})
\\ {}&\geq \abs{1 + \Tr_\Gamma\left\{\Pi_f(T) \left( \tilde a_T^\dagger(\mu)\tilde a_T(\mu)
-\tilde a_T^\dagger(-\mu)\tilde a_T(-\mu)\right)\right\}}^2.
\end{align*}
The two inequalities together give the final result \eqref{semibound}.
\end{proof}

Equations  \eqref{=hetero} express the independence from $\vartheta$ of the arithmetic mean, in both cases of elastic and inelastic spectra. Equation  \eqref{semibound} is a bound of Robertson type; such a bound does not depend on $\vartheta$, but it is still dependent on the initial state and on the dynamics. The
Heisenberg-type relation \eqref{Heisenberg} and the bound \eqref{arithbound} are fully independent of the initial state and of the dynamics.

One speaks of \emph{squeezed field} or of the \emph{spectrum of squeezing} \cite[Sec.\ 9.3.2]{Car08} if, at least in a region of the $\mu$ line, for
some $\vartheta$ one has $S_T^{\mathrm{inel}}(\mu;\vartheta)<1$. If this happens, the
bounds \eqref{Heisenberg} and \eqref{arithbound} say that necessarily
$S_T^{\mathrm{inel}}(\mu;\vartheta+\frac \pi 2)>1$ in such a way that the product and the arithmetic mean are
bigger than one. Note that, with our choice of the environment initial state, any possible squeezing can be imputed to the interaction with $S$. Indeed, in the case of no interaction, the output $W_1$ is a Wiener process plus a deterministic drift and $S_T^{\rm inel} (\mu,\vartheta)\equiv 1$.

\section{Homodyning versus heterodyning of the fluorescence light of a two-level atom}\label{sec:model}
Let us take as system $S$ a two-level atom, which means $\Hscr=\Cbb^2$, $H_0=\frac
{\nu_0} 2\, \sigma_z$; $\nu_0>0$ is the \emph{resonance frequency} of the atom. We
denote by $\sigma_-$ and $\sigma_+$ the lowering and rising operators and by
$\sigma_x=\sigma_-+\sigma_+$, $\sigma_y=\rmi(\sigma_--\sigma_+)$, $\sigma_z=\sigma_+\sigma_- -
\sigma_-\sigma_+$ the Pauli matrices; we set also \[ \sigma_\vartheta =
\rme^{\rmi\vartheta}\,\sigma_- + \rme^{-\rmi\vartheta}\,\sigma_+ , \qquad P_\pm= \sigma_\pm \sigma_\mp. \]
The atom can absorb and emit light and it is stimulated by a laser; some thermal environment can be present too. The
quantum fields $\Gamma$ model the whole environment.

\paragraph{The absorption/emission terms.} The electromagnetic field is split in two
fields, according to the direction of propagation: one field for the photons in the forward
direction ($k=2$), that of the stimulating laser and of the lost light, one field for the
photons collected by the detector ($k=1$). In the rotating wave approximation we can take
\[
R_1=\sqrt{\gamma p}\, \sigma_-\,, \qquad R_2=\sqrt{\gamma
(1-p)}\, \sigma_-\,.
\]
The coefficient $\gamma>0$ is the natural \emph{line-width} of the atom, $p$ is the fraction
of the detected fluorescence light and $1-p$ is the fraction of the lost light
($0<p<1$) \cite{Car93,GarZ00,Bar06,BarGL09}.

\paragraph{Other dissipation terms.} We introduce also the interaction with a thermal
bath,
\[
R_3=\sqrt{\gamma \overline{n}}\, \sigma_-\,, \qquad R_4=\sqrt{\gamma \overline{n}}\,
\sigma_+\,, \qquad \overline{n}\geq 0,
\]
and a term responsible of \emph{dephasing} (or decoherence),
\[
R_5=\sqrt{\gamma k_d}\, \sigma_z\,, \qquad k_d\geq 0.
\]

\paragraph{The laser wave.} We consider a perfectly coherent monochromatic laser of frequency $\nu>0$:
\begin{equation}\label{laserw}
f_k(t)=
\delta_{k2}\, \frac{\rmi \Omega} {2\sqrt{\gamma (1-p)}}\, \rme^{-\rmi \nu t}1_{[0,T]}(t);
\end{equation}
$T$ is a time larger than any other time in the theory and the limit $T\to +\infty$ is taken
in all the physical quantities. The quantity $\Omega\geq 0$ is called \emph{Rabi frequency}
and $\Delta\nu=\nu_0-\nu$ is called \emph{detuning.}

\paragraph{Master equation.} With these choices the Liouville operator \eqref{Lop} becomes
\begin{align}\notag
\Lcal(t)[\rho]= &-\frac \rmi 2 \left[ \nu_0\sigma_z +\Omega \sigma_{\nu t },\, \rho \right] + \gamma k_d \left(\sigma_z\rho\sigma_z -\rho\right)
\\  \notag {}&+ \gamma\left( \overline{n} + 1 \right) \left(\sigma_-\rho\sigma_+ -\frac 1 2 \left\{P_+,\rho\right\}\right)
\\ {}&+{}
\gamma\overline{n}  \left(\sigma_+\rho\sigma_- -\frac 1 2 \left\{P_-,\rho\right\}\right) .
\end{align}
The master equation \eqref{masteq} can be solved by using Bloch equations in the rotating frame \cite[Sec.\ 8.2]{BarG09}.
Indeed, we have
\begin{subequations}\label{rotating_f}
\begin{equation}
\Upsilon(t,s)[\rho]=\rme^{-\frac \rmi 2\, \nu t \sigma_z}\rme^{\check \Lcal(t-s)}\left[\rme^{\frac \rmi 2\, \nu s \sigma_z}\rho \rme^{-\frac \rmi 2\, \nu s \sigma_z}\right]\rme^{\frac \rmi 2\, \nu t \sigma_z},
\end{equation}
\begin{align}\notag
\check\Lcal[\rho]= &-\frac \rmi 2 \left[ \nu_0\sigma_z +\Omega \sigma_{x },\, \rho \right]  + \gamma k_d \left(\sigma_z\rho\sigma_z -\rho\right) \\ \notag {}&+ \gamma\left( \overline{n} + 1 \right) \left(\sigma_-\rho\sigma_+ -\frac 1 2 \left\{P_+,\rho\right\}\right)
\\ {}&+{}
\gamma\overline{n}  \left(\sigma_+\rho\sigma_- -\frac 1 2 \left\{P_-,\rho\right\}\right).
\end{align}
\end{subequations}
The system reduced state turns out to be given by
\begin{align}\notag
\eta_t&=\frac 1 2 \bigl\{\bbone +\left[x(t)+\rmi y(t) \right]\rme^{\rmi \nu t}\sigma_-
\\ {}&+ \left[x(t)-\rmi y(t) \right]\rme^{-\rmi \nu t}\sigma_+ + z(t) \sigma_z \bigr\},
\end{align}
where
\[
\vec x(t) =\rme^{-At}\vec x (0)-\gamma \,\frac {1-\rme^{-At}}A\begin{pmatrix} 0\\ 0 \\ 1 \end{pmatrix},
\]
\[
A=\begin{pmatrix}\gamma\left(\frac{1}{2}+\overline{n}+2k_\rmd \right)& \Delta \nu&0\\
-\Delta \nu &\gamma\left(\frac{1}{2}+\overline{n}+2k_\rmd \right)&\Omega\\
0&-\Omega&\gamma\left(1+2\overline{n} \right)\end{pmatrix}.
\]

\subsection{Homodyning}
The squeezing in the
fluorescence light is revealed by homodyne detection, which needs to maintain phase coherence
between the laser stimulating the atom and the laser in the detection apparatus which
determines the observables $Q(t;\vartheta,h)$.
To maintain phase coherence for a long time, the stimulating wave $f$ and the local oscillator wave $h$ must be produced by the same physical source and this means to take $h$ proportional to $f$. So, by including any phase shift in the pase $\vartheta$ already present in the definition \eqref{quadrature}, to describe homodyning we take
\begin{equation}\label{lo_homo}
h(t)=\frac{-\rmi f_2(t)}{\abs{f_2(t)}}.
\end{equation}
With the choice \eqref{quadrature} for $f$, we get $h(t)=\rme^{-\rmi \nu t}1_{[0,T]}(t)$.

The limit $T\to +\infty$ can be taken in Eqs.\ \eqref{spectrumS} and it is independent of the
atomic initial state \cite{BarGL09}. The result is \cite[Sec.\ 9.2.1]{BarG09}
\begin{align}\notag
S_\mathrm{hom}^{\mathrm{el}}(\mu;\vartheta):&=\lim_{T\to + \infty} S_T^{\mathrm{el}}(\mu;\vartheta)\\ {}&=2
\pi \gamma p \abs{\vec s(\vartheta) \cdot \vec x_\mathrm{eq}}^2  \delta(\mu),
\end{align}
\begin{align}\notag
S_\mathrm{hom}^{\mathrm{inel}}(\mu;\vartheta):&= \lim_{T\to + \infty}
S_T^{\mathrm{inel}}(\mu;\vartheta)\\ {}&=1+2p\gamma
\, \vec{s}(\vartheta)\cdot\left(\frac{A}{A^2+\mu^2}\,\vec{t}(\vartheta)\right),
\end{align}
where
\begin{equation*}
\vec{t}(\vartheta)=\begin{pmatrix} \left(1+z_\mathrm{eq}-x_\mathrm{eq}^{\;2}\right)\cos \vartheta -x_\mathrm{eq}y_\mathrm{eq} \sin \vartheta \\
\left(1+z_\mathrm{eq}-y_\mathrm{eq}^{\;2}\right)\sin \vartheta -x_\mathrm{eq}y_\mathrm{eq} \cos \vartheta
\\
- \left(1+z_\mathrm{eq}\right)\vec s(\vartheta) \cdot \vec x_\mathrm{eq}
\end{pmatrix},
\end{equation*}
\[
\vec{s}(\vartheta)=\begin{pmatrix}\cos\vartheta\\ \sin\vartheta\\0\end{pmatrix}, \qquad
\vec{x}_\mathrm{eq} = -\gamma A^{-1}\,
\begin{pmatrix}0\\0\\1\end{pmatrix}.
\]

Examples of inelastic homodyne spectra are plotted in Figure \ref{BGfig1} for $\gamma=1$, $\overline{n}=k_\rmd=0$, $p=4/5$.
The Rabi frequency $\Omega$, the detuning $\Delta\nu$ and the phase $\vartheta$ are
chosen in order to get the deepest minimum of $S^{\mathrm{inel}}_\mathrm{hom}$ in $\mu=2$. Thus, in this
case the analysis of the homodyne spectrum reveals the squeezing of the detected light. Also the
complementary spectrum is shown, in order to illustrate the role of the Heisenberg-type uncertainty relation \eqref{Heisenberg}. One could also compare the homodyne
spectrum with and without $\overline n$ and $k_\rmd$, thus verifying that the squeezing is
very sensitive to any small perturbation.
\begin{figure}[h]
\begin{center}
\includegraphics*[scale=.25]{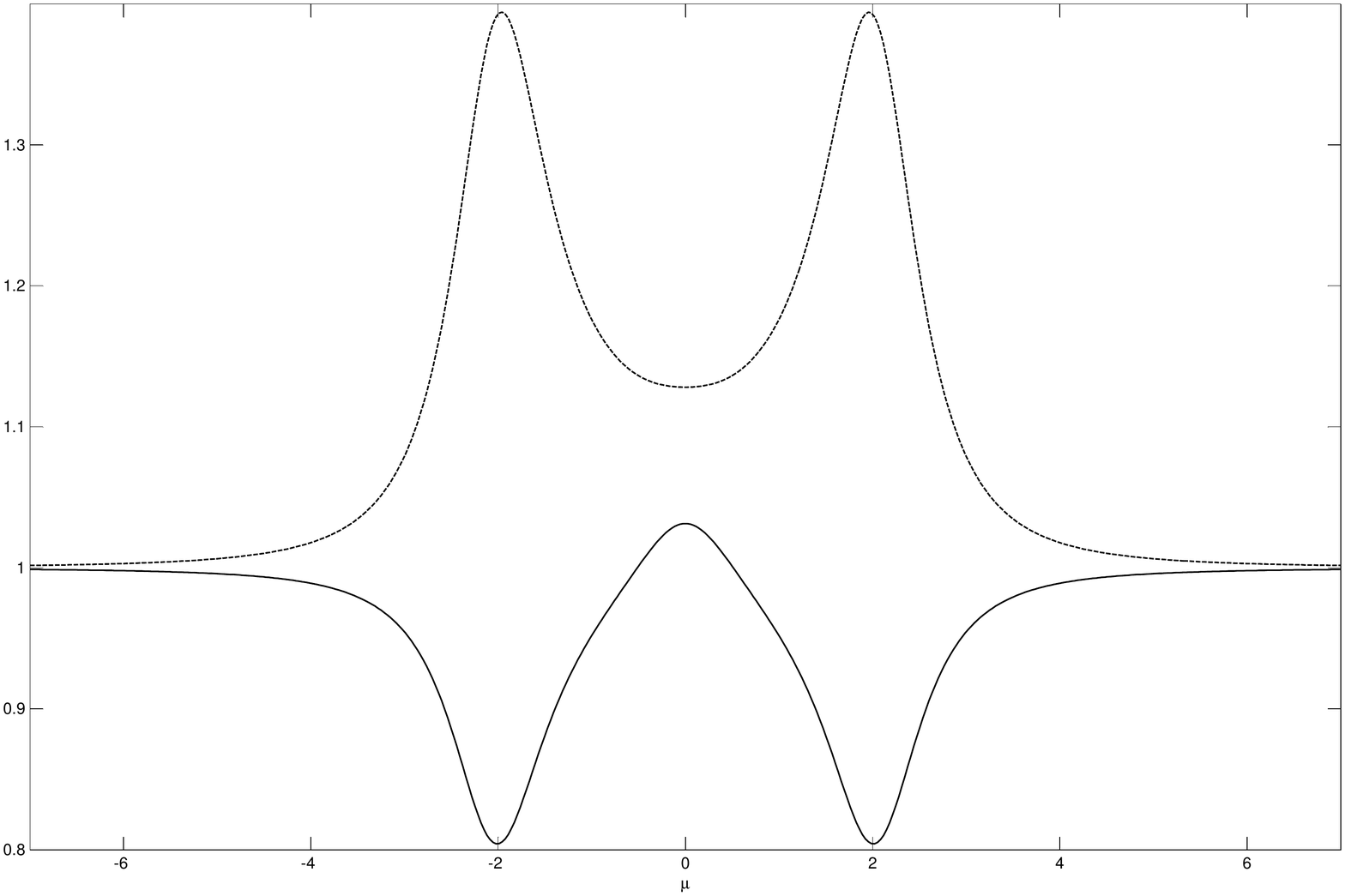} \caption{$S^{\mathrm{inel}}_{\mathrm{hom}}(\mu;\vartheta)$ \
with \ $\Delta\nu=1.4937$, \ $\Omega=1.4360$ and $\vartheta=-0.1748$ (solid line), $\vartheta=\frac \pi 2 -0.1748$ (dashed line).} \label{BGfig1}
\end{center}
\end{figure}

\subsection{Heterodyning} When the local oscillator and the stimulating wave are not produced by the same source, the phase difference cannot be maintained for a long time; in this case we have heterodyne detection. In the case of perfectly monochromatic waves, with a stimulating laser represented by \eqref{laserw}, we take as local oscillator
\begin{equation}\label{lo_het}
h(t)=\rme^{-\rmi \nu_\mathrm{lo} t}1_{[0,T]}(t), \qquad \nu_\mathrm{lo}\neq \nu.
\end{equation}

Again the limit $T\to +\infty$ can be taken in Eqs.\ \eqref{spectrumS} and it turns out to be independent of the
atomic initial state and of $\vartheta$. Let us set
\[
\mathpzc{v}:= \nu_\mathrm{lo}-\nu;
\]
then, the final result is: $S_\mathrm{het}(\mu;\nu_\mathrm{lo})=S_\mathrm{het}^{\mathrm{el}}(\mu;\nu_\mathrm{lo}) + S_\mathrm{het}^{\mathrm{inel}}(\mu;\nu_\mathrm{lo})$,
\begin{subequations}\label{hetSP}
\begin{align}\notag
S_\mathrm{het}^{\mathrm{el}}(\mu;\nu_\mathrm{lo})&=\lim_{T\to + \infty} S_T^{\mathrm{el}}(\mu;\vartheta) \\ \notag
{}&=\frac
\pi 2\,  \gamma p \left( x_\mathrm{eq}^{\,2}+ y_\mathrm{eq}^{\,2}\right)\bigl( \delta(\mu-\mathpzc{v}) + \delta(\mu+\mathpzc{v})\bigr)
\\ \notag {}&=
\frac 1 4 \bigl[S_\mathrm{hom}^{\mathrm{el}}(\mu-\mathpzc{v};0) + S_\mathrm{hom}^{\mathrm{el}}(\mu-\mathpzc{v};\pi/2) \\ {}&+ S_\mathrm{hom}^{\mathrm{el}}(\mu+\mathpzc{v};0) + S_\mathrm{hom}^{\mathrm{el}}(\mu+\mathpzc{v};\pi/2)\bigr],
\end{align}
\begin{align}\notag
S_\mathrm{het}^{\mathrm{inel}}&(\mu;\nu_\mathrm{lo})= \lim_{T\to + \infty}
S_T^{\mathrm{inel}}(\mu;\vartheta)= \gamma p\,D(\mu,\mathpzc{v}) \\ \notag {}&+
\frac 1 4 \bigl[S_\mathrm{hom}^{\mathrm{inel}}(\mu-\mathpzc{v};0) + S_\mathrm{hom}^{\mathrm{inel}}(\mu-\mathpzc{v};\pi/2)  \\ {}&+ S_\mathrm{hom}^{\mathrm{inel}}(\mu+\mathpzc{v};0)+ S_\mathrm{hom}^{\mathrm{inel}}(\mu+\mathpzc{v};\pi/2)\bigr],
\end{align}
\begin{align}\notag
&D(\mu,\mathpzc{v}):={} \\ \notag {}&=\vec s\left ({\textstyle \frac\pi 2}\right)\cdot \left( \frac{(\mu+\mathpzc{v})/2}{A^2+(\mu+\mathpzc{v})^2}-  \frac{(\mu-\mathpzc{v})/2}{A^2+(\mu-\mathpzc{v})^2} \right) \vec t(0) \\ {}&-{}\vec s (0)\cdot \left( \frac{(\mu+\mathpzc{v})/2}{A^2+(\mu+\mathpzc{v})^2}-  \frac{(\mu-\mathpzc{v})/2}{A^2+(\mu-\mathpzc{v})^2} \right) \vec t\left({\textstyle \frac\pi 2}\right).\label{inelhet}
\end{align}
\end{subequations}
Recall that $\vec s (0)=(1,0,0)$ and $\vec s (\pi/2)=(0,1,0)$.

The inelastic heterodyne spectrum \eqref{inelhet} can also be written as
\begin{equation}
S_\mathrm{het}^{\mathrm{inel}}(\mu;\nu_\mathrm{lo})=1+2\pi p \left[\Sigma_{\mathrm{inel}}(\mu+\mathpzc{v}) + \Sigma_{\mathrm{inel}}(\mu-\mathpzc{v})\right],
\end{equation}
\begin{equation}\label{fluor}
\Sigma_{\mathrm{inel}}(\mu)=\frac \gamma {4\pi}\, \RE\left( (1,\rmi, 0) \cdot \frac 1 {A+\rmi \mu}\left[\vec t (0) -\rmi \vec t(\pi/2)\right]\right).
\end{equation}

\subsubsection{Properties of the heterodyne spectrum}
By explicit computations, it is possible to prove \cite[Proposition 9.3 and Remark 9.4 in Sec.\ 9.1.2]{BarG09} that
\[
\Delta \nu=0 \quad  \Rightarrow \quad D(\mu,\mathpzc{v})=0
\]
and that
\[
\overline n =0 \ \text{ and } \ k_d=0 \quad \Rightarrow \quad D(\mu,\mathpzc{v})=0.
\]
In these cases the heterodyne spectrum reduces to a linear combination of different homodyne contributions.

\paragraph{The lower bound of the heterodyne spectrum.}
Being $S_\mathrm{het}^{\mathrm{inel}}(\mu;\nu_\mathrm{lo})$ independent of $\vartheta$, any one of the two bounds in Theorem \ref{Th1} implies
\begin{equation}
S_\mathrm{het}^{\mathrm{inel}}(\mu;\nu_\mathrm{lo})\geq 1, \qquad \forall \mu, \quad \forall \nu_\mathrm{lo}.
\end{equation}
This means that it is impossible to see squeezing by heterodyning. This is true not only in the model of this section, but in any physical set up for which the dependence on $\vartheta $ is lost.

\subsubsection{The power spectrum}
Let us consider now the heterodyne spectrum as a function of the frequency of the local oscillator in the case $\mu=0$. By particularizing the expressions \eqref{hetSP}, we get
\begin{equation*}
S_\mathrm{het}^{\mathrm{el}}(0;\nu_\mathrm{lo})=
\pi   \gamma p \left( x_\mathrm{eq}^{\,2}+ y_\mathrm{eq}^{\,2}\right)\delta(\mathpzc{v}),
\end{equation*}
\begin{align}\notag
S_\mathrm{het}^{\mathrm{inel}}(0;\nu_\mathrm{lo})&=
\frac 1 2 \bigl(S_\mathrm{hom}^{\mathrm{inel}}(\mathpzc{v};0) + S_\mathrm{hom}^{\mathrm{inel}}{\textstyle\left(\mathpzc{v};\frac \pi 2\right)}\bigr)+ \gamma p\,D(0;\mathpzc{v})
\\ {}&=
1+4\pi p \Sigma_{\mathrm{inel}}(\mathpzc{v}),
\end{align}
\begin{equation*}
D(0,\mathpzc{v})=\vec s {\textstyle\left(\frac \pi 2 \right)}\cdot \left( \frac{\mathpzc{v}}{A^2+\mathpzc{v}^2} \,\vec t(0) \right)
-\vec s (0)\cdot \left( \frac{\mathpzc{v}}{A^2+\mathpzc{v}^2}\, \vec t{\textstyle\left(\frac \pi 2\right)} \right).
\end{equation*}

It can be shown that $S_\mathrm{het}(0;\nu_\mathrm{lo})$ is proportional to the mean of the power of the heterodyne current \cite[Sec.\ 9.1.1]{BarG09}; so, as a function of $\nu_\mathrm{lo}$, it represents the \emph{power spectrum}. This interpretation can be strengthened by expressing $S_\mathrm{het}(0;\nu_\mathrm{lo})$ in terms of the fields.
Let us consider now the mode operators \eqref{a(mu)} in the case $h(t)=1$:
\[
a_T(\mu)\big|_{h=1}\equiv \frac 1 {\sqrt{T}} \int_0^T \rme^{\rmi \mu t} \, \rmd B_1(t)=: \hat a_T(\mu).
\]
These operators, together with their adjoint, satisfy the bosonic commutation relations \eqref{muCCR}. By taking into account that the $\vartheta$-dependent terms vanish in the limit $T\to \infty$, from \eqref{thetadependence} we get
\begin{equation}
S_\mathrm{het}(0;\nu_\mathrm{lo})=1 + 2\lim_{T\to +\infty}\Tr_\Gamma\left\{\Pi_f(T) \hat a_T^\dagger(\nu_\mathrm{lo})\hat a_T(\nu_\mathrm{lo})\right\}.
\end{equation}
So, the mean observed power spectrum is composed by the flat spectrum of the shot noise, conventionally set equal to 1, plus a term proportional to the mean number of photons in the temporal mode of frequency approximately equal to $\nu_\mathrm{lo}$.

\paragraph{The fluorescence spectrum} The quantity
\[
\frac{S_{\mathrm{het}}(0;\nu_{\mathrm{lo}})-1}{4\pi p}=\frac{\gamma \left( x_\mathrm{eq}^{\,2}+ y_\mathrm{eq}^{\,2}\right)}4\, \delta(\mathpzc{v})+\Sigma_{\mathrm{inel}}(\mathpzc{v})
\]
is interpreted as the \emph{fluorescence spectrum} of the atom \cite[Sec.\ 9.1.2]{BarG09}; the normalization is chosen in order to have its integral over $\mathpzc{v}$ equal to the rate of emission of photons in the equilibrium state. For $\overline{n}=k_d=0$ this quantity coincides with the original Mollow spectrum \cite[Sec.\ 9.1.2.2]{BarG09}, \cite[pp.\ 178--181]{WisM10}, \cite[p.\ 288]{GarZ00}. In Figure \ref{BGfig2} we give the plot of the inelastic part of the fluorescence spectrum in an  asymmetric case ($\overline n \neq 0$, $k_\rmd\neq 0$) in which the Mollow triplet is well visible ($\Omega$ large).
\begin{figure}[h]
\begin{center}
\includegraphics*[scale=.25]{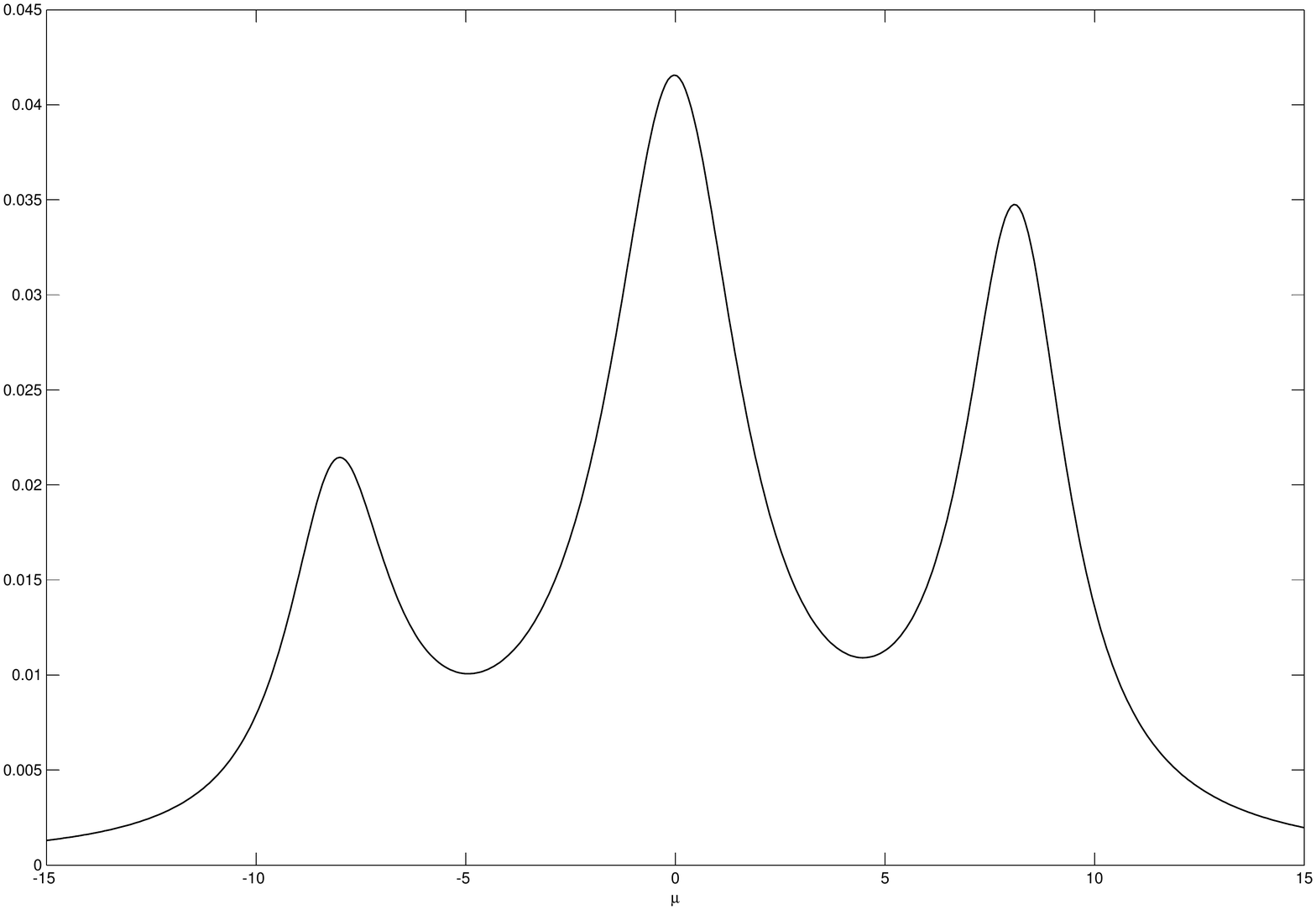} \caption{$\Sigma_{\mathrm{inel}}(\mu)$ \
with \ $\Delta\nu=1.5$, \ $\Omega=8.0$, $\overline n=0.01$, $k_\rmd=0.7$.} \label{BGfig2}
\end{center}
\end{figure}

\section{Discussion and conclusions}\label{sec:concl}

In this paper we have presented and connected two formulations of the quantum theory of measurements in continuous time and quantum filtering. The first one is based on the Hudson-Parthasarathy equation, on quantum stochastic calculus and on the observation of commuting field observables. The other one is based on the SSE, the stochastic master equation and the notion of conditional states. Then, we have studied the properties of the observed output and in particular the output spectrum and the uncertainty relations on the spectra of incompatible quadratures.

 Being the two formulations of quantum continuous measurements equivalent, the output spectra can be deduced from the classical SSE or from the quantum one, and the bounds (of Heisenberg type) on the spectra hold independently of the formulation. This point is conceptually relevant, because this equivalence and the presence of uncertainty relations in both formulations show that the ``classical'' SSE is not a semiclassical approximation to some ``quantum'' theory, but it is itself fully ``quantum''. However, let us stress that the proof of the bounds is based on the quantum field formulation and on the uncertainty relations for \emph{ incompatible quadratures} of the fields. On the other side the most complete results on the probabilistic structure of the output is obtained in the stochastic formulation.

The second part of the paper is devoted to a concrete application (a two-level atom), which allows to discuss  the differences between homodyning and heterodyning and to show how to introduce typical dissipative effects. We show also how to deduce from the general theory the spectrum of the squeezing and the power spectrum with the Mollow triplet.

Various generalizations of the theory presented in this paper are possible, first of all by introducing direct detection \cite{BarP02,Dav76,BarL85,BarB91,Car93,Car08,GarZ00,Bar06,Bar90QO,BarL00,Bar97} and Markovian feedback \cite{BarG09,WanWM01,WisMil93,WisM93,BarGL08,BarGL09,Gough12,DSHB12,Bel12}. Let us stress that closed loop control \cite{BelE08,BouH08,James08,James11,AltT12,TicNA13}, based on the observation of the system, is possible only inside a theory allowing to describe continuous monitoring of a quantum system.

The quantum trajectory approach can be generalized to introduce also feedback control with delay, coloured noises, various non-Markovian effects \cite{BarPP10,BarP10,BarDPP11,BarG12,BarPP12,GJN12}. The simplest non-Markovian contribution is to take the laser wave $f$ to be random
\cite{BarP02, Bar06}. This means to consider as initial state of the field a classical mixture of coherent states, the mean of $\varrho_\Gamma(f)$. Even the local oscillator wave $h$ can be random, which again means to take a mixture of coherent vectors as state of the local oscillator \cite{Bar06,BarPP12}. This randomness can be used to introduce more realistic models of laser light, not only perfectly coherent monochromatic waves, but also waves exhibiting some coherence time. This allows for a better analysis of the  differences between homodyning and heterodyning  \cite{BarPP12}.
Another type of generalization is to consider  cascades systems \cite[Chapt.\ 12]{GarZ00} and networks of optically active systems \cite{Gough12, Mab+12, NG12,EGJ12}. Here, the general Hudson-Parthasarathy equation and the input/output formalism of Eqs.\eqref{out1}--\eqref{Bout} are essential.

\end{document}